\RequirePackage[l2tabu, orthodox]{nag}

\documentclass[prx,aps,superscriptaddress,twocolumn,nopacs,letter,nofootinbib,longbibliography,notitlepage]{revtex4-1}

\usepackage{amsmath,amsfonts,amsthm,xcolor,bbm,hyperref,mathtools,ragged2e,stmaryrd,soul,hhline}
\usepackage[T1]{fontenc}
\usepackage[capitalize, noabbrev]{cleveref}

\usepackage{tikz,pgfplots}
\usetikzlibrary{arrows,calc,decorations.pathmorphing,external,snakes}
\DeclareMathOperator{\Tr}{Tr}

\newcounter{thm}
\newtheorem{theorem}[thm]{Theorem}

\newcounter{lem}
\newtheorem{lemma}[lem]{Lemma}

\newcounter{def}
\theoremstyle{definition}
\newtheorem{definition}[def]{Definition}

\definecolor{darkblue}{RGB}{0,0,128} % choose dark colors for high contrast
\definecolor{darkgreen}{RGB}{0,150,0}
\hypersetup{breaklinks, colorlinks, linkcolor=blue, citecolor=darkgreen, filecolor=red, urlcolor=darkblue}

\newcommand{\comm}[2]{\left\llbracket#1,#2\right\rrbracket}

\usepackage{amsmath}
\DeclareMathOperator*{\argmax}{arg\,max}
\DeclareMathOperator*{\argmin}{arg\,min}

\begin{document}
\title{Statistical mechanical models for quantum codes with correlated noise}
\author{Christopher T.\ Chubb}
\affiliation{Centre for Engineered Quantum Systems, School of Physics,
	University of Sydney, Sydney NSW, Australia}
\author{Steven T.\ Flammia}
\affiliation{Centre for Engineered Quantum Systems, School of Physics,
	University of Sydney, Sydney NSW, Australia}
\affiliation{Yale Quantum Institute, Yale University, New Haven, Connecticut 06520, USA}
\date{\today}

\begin{abstract}
We give a broad generalisation of the mapping, originally due to Dennis, Kitaev, Landahl and Preskill, from quantum error correcting codes to statistical mechanical models. 
We show how the mapping can be extended to arbitrary stabiliser or subsystem codes subject to correlated Pauli noise models, including models of fault tolerance. 
This mapping connects the error correction threshold of the quantum code to a phase transition in the statistical mechanical model.
Thus, any existing method for finding phase transitions, such as Monte Carlo simulations, can be applied to approximate the threshold of any such code, without having to perform optimal decoding.
By way of example, we numerically study the threshold of the surface code under mildly correlated bit-flip noise, showing that noise with bunching correlations causes the threshold to drop to $p_{\textrm{corr}}=10.04(6)\%$, from its known iid value of $p_{\text{iid}}=10.917(3)\%$.
Complementing this, we show that the mapping also allows us to utilise any algorithm which can calculate/approximate partition functions of classical statistical mechanical models to perform optimal/approximately optimal decoding. 
Specifically, for 2D codes subject to locally correlated noise, we give a linear-time tensor network-based algorithm for approximate optimal decoding which extends the MPS decoder of Bravyi, Suchara and Vargo. 
\end{abstract}
\maketitle

%\tableofcontents

% Section 0
Quantum mechanical systems are inherently sensitive to noise. 
The inability to completely suppress environmental noise and perform noiseless quantum operations therefore provides a significant barrier to scalable quantum information processing. 
To mitigate this, quantum error correcting codes~\cite{Shor1995,Steane1996} were developed that encode quantum information into a larger system whose redundant degrees of freedom provide protection from physical noise. 
For a given code family and noise process, quantum information can be encoded and decoded arbitrarily well below a critical noise strength known as the \emph{threshold}. Whilst the threshold is defined with respect to the optimal decoder, it is also often also studied for the case of specific sub-optimal decoding procedures.

The most commonly studied model of noise in quantum codes is that of \emph{iid Pauli noise}, in which each qubit is subjected to an independent, identical Pauli noise process.
This is a mathematically convenient model, but it does not account for any possible correlations between errors. Whilst this model can provide a proof-of-principle that an error correction procedure can successfully withstand errors, many of the physical architectures in which we might hope to implement quantum error correction are known to experience correlated noise. 
Examples include proximity effects~\cite{Fowler2014,Novais2017} and bosonic couplings~\cite{Brown2016a,Alicki2009,Chesi2010,Novais2013,Jouzdani2013,Freeman2014,Novais2017,Lopez-Delgado2017} in solid state systems. 
Correlations can also arise when modelling non-Markovian noise processes~\cite{Ng2009,McCutcheon2014}. 
As the threshold depends on the error model, including the presence and magnitude of correlations, taking these factors into account is important when attempting to specify physically relevant thresholds.

One of the most important correlated noise models is that of \emph{circuit-based noise}~\cite{Raussendorf2006,Fowler2009,Nickerson2013,Tomita2014,Fowler2012,Brown2016b}, in which elementary gates and measurement are taken to be noisy.
Even if we assume the noise introduced in each operation is independent, the operations themselves tend to propagate and accumulate errors, giving an overall correlated noise model. 
Importantly, there exist codes which have a threshold under this correlated noise model, which is known as the \emph{fault-tolerant threshold}~\cite{Kitaev1997,Knill1996}. 
Moreover, it can even be shown that there exist error correction procedures which allow not only for fault tolerant storage of quantum information, but also fault tolerant quantum computation, a result known as the \emph{quantum threshold theorem}~\cite{Knill1998a,Aharonov1999,Aharonov2006,Preskill2013}.

In the context of correlated noise, little is known about optimal decoding or fault tolerance procedures. For this reason, most Monte Carlo estimates of thresholds for correlated noise are given with respect to sub-optimal decoders and fault tolerance schemes~\cite{Baireuther2018,Krastanov2017,Wootton2012,Criger2017,Varsamopoulos2018,Delfosse2014,Hutter2014,Duclos-Cianci2010,Nickerson2017,Maskara2018,Darmawan2018}. Indeed, to our knowledge, no optimal thresholds are known for any interesting quantum code families with a non-trivially correlated noise models. 

Remarkably, there exists a method for computing code thresholds with respect to iid Pauli noise called the \emph{statistical mechanical mapping}~\cite{Dennis2001}. 
In this technique, a classical statistical mechanical system is constructed from a quantum code with the noise model manifesting as a quenched disorder. 
This mapping is designed in such a way that the thermodynamic properties of the statistical mechanical system relate to the error correction properties of the quantum code, under the optimal decoder.
This method has been used~\cite{Fujii2013} to compute thresholds~\cite{Bombin2012,Kovalev2013,Katzgraber2009,Kovalev2013, Katzgraber2009, Bombin2012, Kubica2017,Kovalev2018}, including fault-tolerance thresholds~\cite{Wang2003,Andrist2011a,Dumer2015,Bombin2013,Andrist2012, Ohno2004}, for a wide variety of code families based on topological codes~\cite{Bombin2013}, i.e.\ codes with spatially local stabiliser generators on a lattice in Euclidean space. 

Importantly, this link implies that the error correction threshold in the code manifests as a phase transition in the resulting stat mech model. 
This insight implies that one may bring to bear the various numerical and analytical techniques for determining the phases of stat mech systems to indirectly estimate the threshold of our code, \emph{without} having to implement optimal decoding. 

\subsection*{Summary of main results}

In this manuscript we give a broad generalisation of the stat mech mapping to the case of correlated Pauli noise models acting on any stabiliser or subsystem quantum code. 
The original mapping for independent bit-flip noise~\cite{Dennis2001}, and subsequent generalisation for independent Pauli noise~\cite{Bombin2012}, works by showing that when a special condition known as the Nishimori condition~\cite{Nishimori1981} holds, a certain stat mech model with quenched disorder has a partition function that maps directly onto the probability of a logical class given the syndrome of a code. 
Our first result is that this fundamental theorem continues to hold in much more generality for \emph{correlated} noise using our more general stat mech mapping. 

The specific notion of correlation allowed by our theorem is very general. 
We first show that any distribution arising from a \textit{factor graph} admits such a mapping, which generalises independent noise as follows. 
In particular, it works whenever any cluster of errors that is sufficiently far apart is \emph{conditionally} independent given the neighbouring spins of the clusters. 
This is the well-known spatial Markov condition, and it provides a systematic relaxation of the notion of independence, controlled by the length scale at which disjoint clusters become conditionally independent. 
When this length scale is zero, there are no intermediate spins on which to condition, and we recover the case of strictly independent noise. 
We make this notion precise in \cref{def:factored} below. 

We further generalise this to the case of spatio-temporally correlated noise by mapping to a system of one higher spatial dimension. 
This allows us to include the most relevant type of noise for quantum computing, namely circuit-based noise thresholds.

By performing Monte Carlo simulations of the resulting stat mech system, we then use this correspondence to approximate the threshold of the toric code subject to a correlated model of bit-flip errors, and quantify how a certain family of positively correlated errors affect the threshold. 

Finally we will show how this mapping can also be used to give an efficient approximation of the maximum likelihood decoder.
A consequence of this result is a generalisation of the tensor network decoder of Ref.~\cite{Bravyi2014} to any 2D surface code with spatially local noise correlations. 
The tensor network that yields the decoder can be approximately contracted in linear time in $n$, the number of qubits. 
The contraction sequence allows a systematic approximation of the maximum likelihood decoder by increasing the bond dimension cutoff of an intermediate matrix product state representation. 

The paper is organised as follows. 
After introducing mathematical preliminaries and definitions in \cref{sec:prelims}, we continue in \cref{sec:mapping} by reviewing the stat mech mapping in the case of independent noise.
In \cref{sec:correlated} we then extend the mapping to account for noise models with spatial correlations. 
In \cref{sec:numerics} we apply this mapping to correlated bit-flip noise in the toric code, and perform Monte Carlo simulations to estimate the threshold. 
In \cref{sec:circuit} we extend our construction to spatio-temporal correlations and circuit-based noise. 
In \cref{sec:decoding} we prove that the phase boundary and the threshold coincide, and using this we show in \cref{sec:mld} how tensor network methods yield efficient approximations of the optimal decoder.
We conclude in \cref{sec:conclusions}, provide some background on correlated noise in \cref{app:graphical}, and details of our numerical simulations in \cref{app:simulation}.

% Section 1: Preliminaries

\section{Preliminaries}
\label{sec:prelims}

Fix a local dimension $d\in\mathbb N$. Let $\omega$ denote the fundamental $d$th root of unity, and let $\mathcal{P}$ denote the set of Pauli operators \emph{modulo phase}. We will associate a copy of the Paulis with each site $i$ for $1\leq i\leq n$, denoting the Paulis supported on site $i$ by $\mathcal P_i$, and the set of global Paulis by $\mathcal{P}^n:=\bigtimes_{i=1}^n\mathcal{P}_i$.

Pauli operators always commute up to phase. As such, we can capture the non-commutativity of the Paulis via the \emph{scalar commutator} $\comm{\cdot}{\cdot}:\mathcal P\times \mathcal P\to\mathbb C$, defined by the relation $AB=\comm{A}{B}BA$. It can also be seen as the normalised trace of the group commutator, 
\begin{align}
\comm{A}{B}:=\frac{1}{d}\Tr\left[A,B\right],
\end{align}
where $\left[A,B\right]:=ABA^{-1}B^{-1}$. We can see that the scalar commutator is a well-defined function on $\mathcal P\times \mathcal P$, as the group commutator is invariant under phases, $\left[e^{i\theta}A,e^{i\phi}B\right]=\left[A,B\right]$. In general the group commutator is not multiplicative\footnote{For a general group, $\left[A,BC\right]=\left[A,B\right]B\left[A,C\right]B^{-1}$. }, but conveniently the scalar commutator of Paulis \emph{is}, in the sense that \mbox{$\comm{A}{BC}=\comm{A}{B}\cdot\comm{A}{C}$}. This also implies that Paulis commute under the scalar commutator, in the sense that $\comm{A}{BC}=\comm{A}{CB}$, which will be important for considering subsystem codes.

\subsection{Stabiliser and Subsystem codes}

The codes we will be considering in this work are Pauli stabiliser codes~\cite{Gottesman1996,Gottesman1997}. A stabiliser code is defined by a subgroup $\mathcal S$ of the Paulis acting on $n$ qu$d$its. This subgroup must be Abelian and have trivial overlap with the centre, $\mathcal S\cap \langle\omega I\rangle=\lbrace I\rbrace$. This group is typically specified in the form of a set of generators $\lbrace S_k\rbrace_k$. The associated code space is given by all states which are stabilised by every element of the stabiliser group, i.e.\ all $|\psi\rangle$ such that $s|\psi\rangle=|\psi\rangle$ for all $s\in\mathcal{S}$.

The logical (Pauli) operators of such a code are given by the centraliser of the stabiliser group, that is the set of all $L\in \mathcal P^n$ such that $\comm{s}{L}=+1$ for all stabilisers $s\in\mathcal{S}$. As the stabilisers act trivially on the code space, logicals which differ just by a stabiliser have identical action on the code space, and we will refer to them as \emph{logically equivalent}.

Suppose that we start in a code state $|\psi\rangle$, and our system suffers an error $E\in \mathcal P^n$, leaving it in state $E|\psi\rangle$. The first step in the error correction procedure begins by measuring each of the stabiliser generators $\lbrace S_k\rbrace_k$. The Pauli error model has the key feature that any given error yields a deterministic outcome. Specifically, measuring the $k$th stabiliser generator deterministically gives the outcome $\comm{S_k}{E}$, and therefore does not disturb the state. Importantly, this outcome does not depend on the initial code state $|\psi\rangle$, depending only on the error $E$. We refer to the collection of all such measurement outcomes as the \emph{syndrome}.

A decoder for such a code is an algorithm which takes as input the syndrome, and outputs a decoding Pauli $D$. We then apply $D^{-1}$ to the state, in the hope that this corrects the error. This decoding successfully restores the system back to the original code state if and only if $D$ was logically equivalent to the true error $E$, i.e.\ $D^{-1}E\in \mathcal{S}$.

A more general notion of quantum code is that of a subsystem code~\cite{Kribs2005}. 
Here, some fraction of the logical qubits are sacrificed to become so-called gauge qubits.
Gauge qubits then provide a workspace that simplifies some of the measurements, since instead of measuring a stabiliser directly one can measure combined gauge and stabiliser operators that might have lower weight. 
In a subsystem code, the gauge group is generated by all of the stabiliser and the gauge operators, and the stabiliser group is the centre of the gauge group (modulo phase). 
Thus, measuring enough gauge generators to reconstruct the stabilisers is sufficient to perform quantum error correction in a subsystem code. 
This is true even though general elements of the gauge group don't commute, since the stabiliser elements \emph{do} commute by virtue of living in the centre. 
We refer the reader to Ref.~\cite{Kribs2005} for a more detailed discussion of subsystem codes. 

In what follows, we never use the fact that the stabiliser generators commute. 
We only use the fact that Pauli errors form an abelian error algebra, since they commute modulo the scalar commutator. 
That is, given two Pauli operators $P$ and $Q$, the accumulated error on a state is the same if one applies $PQ$ or if one applies $QP$. 
Because of this fact, everything that we derive below also applies to subsystem codes where gauge generators are used in place of stabiliser generators. 
With the exception of \cref{sec:circuit}, we will use the simpler language of stabiliser codes throughout, but with the understanding that the results can easily be generalised to subsystem codes.

% Section 2: Statistical mechanic mapping

\section{The statistical mechanical mapping}
\label{sec:mapping}

Before considering correlated noise, we start by reviewing the case of independent noise, as first considered in Ref.~\cite{Dennis2001}.
Although this material is review, our goal is to write the derivation of the independent case in such a way that the correlated case falls out as naturally as possible.

For notational convenience, we will restrict out attention to stabiliser codes for this section, but we note that the below construction can be naturally extended to subsystem codes by replacing stabiliser generators by gauge generators.

Consider a code given by a set of stabiliser generators $\lbrace S_k\rbrace_{k}$. 
In this section we will be considering an independent Pauli noise model: let $\lbrace p_i\rbrace_{i}$ be probability distributions on $\mathcal{P}_i$ which describe the probability of a Pauli error (independently) occurring at each site $i$. 
The probability of an overall error $E\in\mathcal{P}^n$ is therefore given by
\begin{align}
	\Pr(E)=\prod_{i=1}^{n}p_i(E_i),
\end{align} 
where $E_i$ denotes the action of $E$ within $\mathcal{P}_i$.

\subsection{Statistical mechanical model}

We now want to develop a (classical) spin model whose statistical mechanical properties capture the error correction properties of our quantum code, in such a way that that the threshold in the latter naturally corresponds to a phase transition in the former. 

The state space of this model will correspond to the stabiliser group, with the noise model determining the interactions. Specifically, associate a classical spin degree of freedom $c_k$ with each stabiliser generator $S_k$. We will consider each degree of freedom as a member of $\mathbb{Z}_d$, and so our full state space\footnote{
	In the case where $d$ is composite, and some of the generators $S_k$ are not of maximal order, this will actually cause an over-counting in \cref{thm:fundie}. Each degree of freedom should take on a number of states equal to the order of $S_k$. For notational convenience we will assume all of our generators $S_k$ are of maximal order, as is always the case for prime $d$.
} is given by $\times_k\mathbb{Z}_d$.

The family of Hamiltonians we will consider is defined as follows.

\begin{definition}[Stat mech Hamiltonian: \emph{independent noise}]
	\label{defn:ham}
	For a Pauli $E\in \mathcal{P}^{\times n}$, and coupling strengths $\lbrace J_i:\mathcal{P}_i\to \mathbb{R}
	\rbrace_i$, our Hamiltonian $H_E$ is defined as
	\begin{align}
	\label{eq:definitionStatMechHam}
	H_E(\vec c)
	=-\sum_{i,\sigma\in \mathcal{P}_i}J_i(\sigma)\comm{\sigma}{E}\prod_k\comm{\sigma}{S_k}^{c_k}\!,
	\end{align}
	for any state $\vec{c}\in\times_k\mathbb{Z}_d$, where $E$ forms a (quenched) disorder parameter.
	Here the sum is over all qudits $i$ and all elements $\sigma$ in the local Pauli group $\mathcal{P}_i$.
\end{definition}
For algebraic manipulation, it will be convenient to note that the above Hamiltonian can also be written as
\begin{align}
H_E(\vec c)
=-\sum_{i,\sigma\in \mathcal{P}_i}J_i(\sigma)\comm{\sigma}{E\prod_kS_k^{c_k}}\!,
\end{align}
due to the multiplicativity of the scalar commutator. 

Physically, these Hamiltonians correspond to random bond Ising-type models for $d=2$. This can be seen by putting the Hamiltonian in terms of the degrees of freedom $s_k:=(-1)^{c_k}$,
\begin{align}
\label{eq:directisingtype}
H_E(\vec s)
=-\sum_{i,\sigma\in \mathcal{P}_i}\overbrace{J_i(\sigma)}^{\text{Strength}}\overbrace{\comm{\sigma}{E}}^{\text{Disorder}}\overbrace{\prod_{k:\comm{\sigma}{S_k}=-1}s_k}^{\text{Interactions}}.
\end{align}
Similarly for $d>2$ this can be seen as a type of random-bond vector Potts model (or clock model).

Readers familiar with the prior work on stat mech Hamiltonians (beginning with~\cite{Dennis2001}) might be puzzled as to why we write such a convoluted form of $H_E$ for our definition rather than the straightforward Ising-type model in \cref{eq:directisingtype}. 
The answer is that, as we will see below, this formulation gives the simplest path to generalising these results to the case of correlated noise.

We note that if our stabiliser generators are local, in the sense that each site is only acted upon non-trivially by a finite number of stabiliser generators, then this Hamiltonian is also local, in the sense that each interaction only touches a finite number of sites. 

In this construction, each of our interactions corresponded to a site $i$ and a $\sigma\in\mathcal{P}_i$, i.e.\ a single-site Pauli. As we will see in \cref{sec:correlated}, this is intimately linked with the restriction to independent noise, and that including larger range interactions will allow us to account for correlated error models.

Before considering specific coupling strengths, we note that this model has been constructed such that it is symmetric under multiplying the disorder $E$ by a stabiliser generator $S_k$ and adding to the corresponding degree of freedom $c_k$. Specifically,
\begin{align}
	H_{ES_k}(\vec c)=H_E(\vec c+\hat k).
\end{align}
This can be seen by noting that our Hamiltonian only depends on each $c_k$ via $ES_k^{c_k}$, which itself has this symmetry.

\subsection{Nishimori condition}

We now want to consider the coupling strengths required such that the above model reproduces the statistical properties of our quantum code and noise model. 
Specifically we want to choose these couplings such that
\begin{align}
	Z_E=\Pr(\overline E),
\end{align}
where $Z_E$ is the partition function of the Hamiltonian $H_E$ with quenched disorder $E$, $\Pr(\overline E)$ denotes the probability of a logical class of errors, and $\overline E$ denotes the set of errors which are logically equivalent to $E$. 
A sufficient condition for this is given by the \emph{Nishimori conditions}.

\begin{definition}[Nishimori conditions: \emph{independent noise}]
	\label{defn:nishimori}
	An inverse temperature $\beta$ and coupling strengths $\lbrace J_i\rbrace_i$ satisfy the \emph{Nishmori conditions} with respect to distributions $\lbrace p_i\rbrace_i$ if
	\begin{align}
	\label{eq:IndependentNishimori}
	\beta J_i(\sigma)=\frac{1}{|\mathcal P|}\sum_{\tau\in \mathcal{P}_i}\log p_i(\tau)\comm{\sigma}{\tau^{-1}},
	\end{align}
	holds for all $i$ and $\sigma\in \mathcal P_i$.
\end{definition}

Note that \cref{eq:IndependentNishimori} is not defined if any of the $p_i(\tau)$ are exactly zero. 
However, in those cases, we can formally manipulate $\log 0$ together with the convention that $\exp(\log 0) = 0$ and obtain sensible answers. 
A more rigorous treatment taking limits is certainly possible, but would obscure the thrust of the argument, so we neglect these details. 

This now allows us to prove the critical property of our stat mech mapping, which we refer to as the \emph{fundamental theorem}.

\begin{theorem}[Fundamental theorem of stat mech mapping: \emph{independent noise}]
	\label{thm:fundie}
	Given the stat mech Hamiltonian \emph{(}\cref{eq:definitionStatMechHam}\emph{)} associated to a noise model that satisfies the Nishimori condition \emph{(}\cref{eq:IndependentNishimori}\emph{)}, the probability of a logical class of errors is equal to the corresponding partition function,
	\begin{align}
	Z_E=\Pr(\bar E).
	\end{align}
\end{theorem}
\begin{proof}
	First we note that the Nishimori condition takes the form of an (inverse) Fourier transform of the log-probabilities $\log p_i$ with respect to the Pauli group. Using the orthogonality condition $\sum_{\sigma}\comm{\sigma}{\rho}=|\mathcal P|\delta_{\rho,I}$, we can see that this is equivalent to the requirement that the Fourier transform of $\beta J_i$ is $\log p_i$, specifically
	\begin{align}
		\sum_{\sigma\in\mathcal{P}_i}\beta J_i(\sigma)\comm{\sigma}{E}=\log p_i(E_i).
	\end{align}
	Summing over sites $i$, we see that this means the Gibbs weight in the all-zero state reproduces the error probability,
	\begin{subequations}
	\begin{align}
		e^{-\beta H_E(\vec 0)}
		&=\exp\left(\sum_{i,\sigma\in\mathcal{P}_i}\beta J_i(\sigma)\comm{\sigma}{E}\right)\\
		&=\exp\left(\sum_i \log p_i(E_i)\right)\\
		&=\prod_i p_i(E_i)
		=\Pr(E).
	\end{align}
	\end{subequations} 
	Lastly, we can use the symmetry property to write the partition function as a sum over logically equivalent disorders, giving the probability of logical class as required,
	\begin{align}
	Z_E
	&=\sum_{\vec c}e^{-\beta H_E(\vec c)}
	=\sum_{S\in\mathcal S}e^{-\beta H_{ES}(\vec 0)}
	=\Pr(\overline E).
	\end{align}
\end{proof}

This correspondence between logical class probabilities and partition functions suggests that the regimes in which error correction is and is not possible in our code correspond naturally to phases in our stat mech model. We will show this explicitly in \cref{sec:decoding}, and see that the error correcting threshold manifests as a quenched phase transition.

In \cref{sec:mld} we will also see how this correspondence, together with tensor network methods for approximating partition functions, can be used to construct a family of efficient tensor network algorithms which approximate maximum likelihood decoding.

\subsection{Example: Toric code}

\begin{figure}
	\centering
	\hspace{-1.5cm}
	\includegraphics{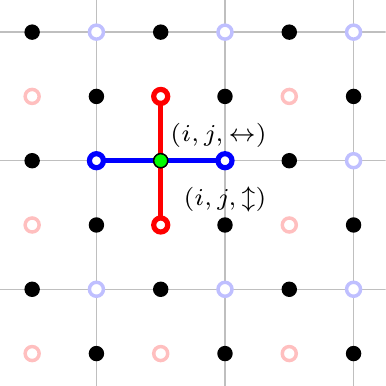}
	\hspace{-.375cm}
	\includegraphics{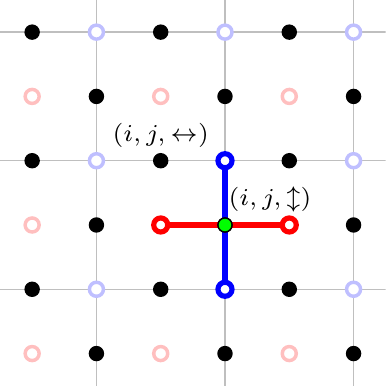}
	\hspace{-1cm}
	
	\caption{Stat mech mapping of toric code with iid noise. Solid circles indicate qubits, and hollow circles indicate the spins of the stat mech model. The lines connecting the spins to the highlighted green qubit indicate the spins which are coupled by the interactions corresponding to that qubit. The interaction terms corresponding to a qubit on a horizontal edge $(i,j,\leftrightarrow)$ give a horizontal Ising coupling on the $X$-sublattice, a vertical Ising coupling on the $Z$-sublattice, and a four-body term coupling the two sub-lattices. The vertical qubit $(i,j,\updownarrow)$ gives rotated versions of these couplings. The full Hamiltonian is given in \cref{eqn:iidtoric1,eqn:iidtoric2}.}
	\label{fig:iidtoric}
\end{figure}

We now consider several examples of the stat mech systems corresponding to the toric code with iid noise. 
Let $p$ be the marginal distribution of errors on a single qubit. 
Given that our stabilisers split into $X$-type stars and $Z$-type plaquettes, we will divide the spins into two corresponding sub-lattices, denoting them $\color{blue}\lbrace {s_k^X}\rbrace_k$, the $X$-sublattice, or $\color{red}\lbrace {s_k^Z}\rbrace_k$, the $Z$-sublattice, respectively. 
Applying our mapping to this code gives the Hamiltonian
\begin{align}
\label{eqn:iidtoric1}
H_E=-\sum_e\Bigl(J(I)
&+J(X)\comm{X}{E_e}{\color{red}\prod_{\partial f\ni e}s^Z_f} \\
&+J(Z)\comm{Z}{E_e}{\color{blue}\prod_{v\in \partial e}s^X_v} \notag\\
&+J(Y)\comm{Y}{E_e}{\color{red}\prod_{\partial f\ni e}s^Z_f} {\color{blue}\prod_{v\in\partial e}s^X_v} \notag
\Bigr),
\end{align}
where $v,e,f$ denote vertices, edges, and faces in the lattice, $\partial e$ denote the vertices surrounding an edge, and $\partial f$ the edges surrounding a face. 
We note that the above Hamiltonian is not only valid for the toric code, but in fact \emph{any} homology code.

Writing this Hamiltonian out more explicitly in Cartesian coordinates, as indicated in \cref{fig:iidtoric}, we see that our model corresponds to two copies of the 2D random bond Ising model on each sublattice, with a four-body coupling between them, specifically
\begin{align}
\label{eqn:iidtoric2}
H_E=-\sum_{i,j}\Bigl(2J(I)
\!
&+\!J(X)\comm{X}{E_{i,j,\leftrightarrow}} {\color{red}s_{i,j}^Zs_{i,j+1}^Z} \\[-.35cm]\notag
&+\!J(X)\comm{X}{E_{i,j,\updownarrow}} {\color{red}s_{i,j}^Zs_{i+1,j}^Z} \\\notag
&+\!J(Z)\comm{Z}{E_{i,j,\leftrightarrow}} {\color{blue}s_{i,j}^Xs_{i-1,j}^X} \\\notag
&+\!J(Z)\comm{Z}{E_{i,j,\updownarrow}}{\color{blue}s_{i,j}^Xs_{i,j-1}^X} \\\notag
&+\!J(Y)\comm{Y}{E_{i,j,\leftrightarrow}} {\color{red}s_{i,j}^Zs_{i,j+1}^Z} {\color{blue}s_{i,j}^Xs_{i-1,j}^X} \\\notag
&+\!J(Y)\comm{Y}{E_{i,j,\updownarrow}}\! {\color{red}s_{i,j}^Zs_{i+1,j}^Z} {\color{blue}s_{i,j}^Xs_{i,j-1}^X} \notag
\Bigr),
\end{align}
where $E_{i,j,\leftrightarrow}$ and $E_{i,j,\updownarrow}$ indicate the action of the error $E$ on the $(i,j)$th horizontal and vertical qubits (see \cref{fig:iidtoric}). The terms inside the parentheses are the interactions corresponding to the qubits labelled $(i,j,\leftrightarrow)$ and $(i,j,\updownarrow)$ as indicated in \cref{fig:iidtoric}.

The Nishimori conditions give coupling strengths of the form
\begin{subequations}
\begin{align}
	J(I)&=\frac{1}{4\beta}\log p(I)p(X)p(Y)p(Z),\\
	J(X)&=\frac{1}{4\beta}\log \frac{p(I)p(X)}{p(Y)p(Z)},\\
	J(Z)&=\frac{1}{4\beta}\log \frac{p(I)p(Z)}{p(X)p(Y)},\\
	J(Y)&=\frac{1}{4\beta}\log \frac{p(I)p(Y)}{p(X)p(Z)}.
\end{align}
\end{subequations}

We now consider this model for specific iid error models of interest.

\subsubsection{Depolarising noise (Bombin et.al.)}

If we consider the depolarising channel, with \mbox{$p(I)=1-p$} and $p(X)=p(Y)=p(Z)=p/3$, then our model reduces to 
\begin{align}
H_E=-\sum_{i,j}\Bigl(-K&+J\comm{X}{E_{i,j,\leftrightarrow}} {\color{red} s_{i,j}^Zs_{i,j+1}^Z}\\[-.35cm]\notag
&+J\comm{X}{E_{i,j,\updownarrow}} {\color{red}s_{i,j}^Zs_{i+1,j}^Z}\\\notag
&+J\comm{Z}{E_{i,j,\leftrightarrow}} {\color{blue}s_{i,j}^Xs_{i-1,j}^X} \\\notag
&+J\comm{Z}{E_{i,j,\updownarrow}} {\color{blue}s_{i,j}^Xs_{i,j-1}^X} \\\notag
&+J\comm{Y}{E_{i,j,\leftrightarrow}} {\color{red}s_{i,j}^Zs_{i,j+1}^Z} {\color{blue}s_{i,j}^Xs_{i-1,j}^X} \\\notag
&+J\comm{Y}{E_{i,j,\updownarrow}} {\color{red}s_{i,j}^Zs_{i+1,j}^Z} {\color{blue}s_{i,j}^Xs_{i,j-1}^X} \notag
\Bigr),
\end{align}
with Nishimori conditions
\begin{subequations}
\begin{align}
\beta K&=\frac{1}{2}\log \frac{27}{p^3(1-p)}\,,\\
\beta J&=\frac{1}{4}\log \frac{3(1-p)}{p}\,.
\end{align}
\end{subequations}
This corresponds to a disordered variant of the eight-vertex model~\cite{Sutherland1970,Fan1970,Baxter1971}, which was considered in Ref.~\cite{Bombin2012}.

\subsubsection{Independent X and Z}

Consider now a model with independent $X$ and $Z$ errors. Specifically let ${\color{blue}p_X},{\color{red}p_Z}$ denote the probability of each \emph{generator}, such that
\begin{subequations}
\begin{align}
	p(I)&={\color{blue}(1-p_X)}{\color{red}(1-p_Z)},\\
	p(X)&={\color{blue}p_X}{\color{red}(1-p_Z)},\\
	p(Z)&={\color{blue}(1-p_X)}{\color{red}p_Z},\\
	p(Y)&={\color{blue}p_X}{\color{red}p_Z}.
\end{align}
\end{subequations}
Importantly, this means that $p(X)p(Z)=p(I)p(Y)$, and therefore $J(Y)= 0$. This has the effect of decoupling the two sub-lattices into two non-interacting random-bond Ising models, such that the Hamiltonian can be decomposed $H_E={\color{blue}H_E^X}+{\color{red}H_E^Z}$, where
\begin{subequations}
\begin{align}
{\color{blue}H_E^X}=-\sum_{i,j}\Bigl(-{\color{blue}K^X}
&+{\color{blue}J^X}\comm{Z}{E_{i,j,\leftrightarrow}} {\color{blue}s_{i,j}^Xs_{i-1,j}^X} \\[-.4cm]\notag
&+{\color{blue}J^X}\comm{Z}{E_{i,j,\updownarrow}} {\color{blue}s_{i,j}^Xs_{i,j-1}^X}~ \notag
\Bigr),\\
{\color{red}H_E^Z}=-\sum_{i,j}\Bigl(-{\color{red}K^Z}
&+{\color{red}J^Z}\comm{X}{E_{i,j,\leftrightarrow}} {\color{red}s_{i,j}^Zs_{i,j+1}^Z} \\[-.4cm]\notag
&+{\color{red}J^Z}\comm{X}{E_{i,j,\updownarrow}} {\color{red}s_{i,j}^Zs_{i+1,j}^Z}~ \notag
\Bigr),
\end{align}
\end{subequations}
and
\begin{subequations}
\begin{align}
K^A&=\frac{1}{2\beta}\log \frac{1}{p_A(1-p_A)}\,,\\
J^A&=\frac{1}{2\beta}\log \frac{1-p_A}{p_A}\,,
\end{align}
\end{subequations}
for $A\in\lbrace X,Z\rbrace$. This corresponds to two decoupled copies of the random-bond Ising model, with disorder probabilities ${\color{blue}p_X}$ and ${\color{red}p_Z}$. The decoupled nature of these two models is generic for CSS codes (i.e., codes whose stabilisers split into separate $X$ and $Z$ type) under independent $X$ and $Z$ noise.
A further consequence of this is that the optimal decoder can decode the $X$ and $Z$ errors independently.

\subsubsection{Pure bit-flip noise (Dennis et.al.)}

As mentioned when we introduced the Nishimori condition above, care must be taken whenever the noise model does not have full support, as our Hamiltonian becomes divergent on certain states. An important example of such a model is that of iid bit-flip noise, as was considered in the seminal paper Ref.~\cite{Dennis2001}. Here we take \mbox{$p(I)=1-p$}, \mbox{$p(X)=p$}, and \mbox{$p(Y)=p(Z)=0$}.  

We know that the Nishimori condition implies
\begin{align}
	H_E(\vec{c})=-\frac{1}{\beta}\log \Pr\left(E\prod_k S_k^{c_k}\right).
\end{align}
Given that the probability of any error containing any non-trivial $Z$ contribution is zero in this error model, this tells us that any state which is \emph{not} entirely magnetised on the $Z$-sub-lattice (${\color{red} s_{i,j}^Zs_{i+1,j}^Z}=-1$ or ${\color{red}s_{i,j}^Zs_{i,j+1}^Z}=-1$ for some $i,j$) has infinite energy. We can interpret this as the degrees of freedom in this sub-lattice being \emph{frozen out}, into the one of the two entirely magnetised states. Restricting to these (degenerate) $Z$-magnetised states (${\color{red}s_{i,j}^Zs_{i+1,j}^Z}={\color{red}s_{i,j}^Zs_{i,j+1}^Z}=+1$ for all $i,j$), our Hamiltonian therefore reduces to a single copy of the random bond Ising model
\begin{align}
H_E=-\sum_{i,j}\Bigl(-2K&+J\comm{Z}{E_{i,j,\leftrightarrow}} {\color{blue}s_{i,j}^Xs_{i-1,j}^X} \\\notag
&+J\comm{Z}{E_{i,j,\updownarrow}} {\color{blue}s_{i,j}^Xs_{i,j-1}^X} \notag\Bigr),
\end{align}
where
\begin{subequations}
\begin{align}
	K&=\frac{1}{2\beta}\log\frac{1}{p(1-p)},\\
	J&=\frac{1}{2\beta}\log \frac{1-p}{p}.
\end{align}
\end{subequations}
We can see now that $J$ and $K$ are both finite for any $p\in(0,1)$, meaning that our Hamiltonian is no longer divergent on the remaining degrees of freedom. Moreover, this is the Hamiltonian that was considered in Ref.~\cite{Dennis2001}. This shows that, taking appropriate care, errors models without full support can also be considered using this construction, after considering frozen out degrees of freedom.

\subsection{Extensions}

The stat mech mapping for independent noise can also be generalised to several other noise models.

\paragraph{Noisy measurements} The first example is a noise model consisting of independent noise, and independently noisy measurements. This can be modelled by including an ancilla bit for each stabiliser, and replacing stabilisers $S\to S\otimes Z$, where the $Z$ acts on this ancilla. Any bit-flip noise on this ancilla bit will effectively model noisy measurements.

\paragraph{Leakage errors} Leakage can be accounted for by explicitly including flag bits.
The precise construction however will depend on what model of leakage is being used.

\paragraph{Overlapping independent} One simple toy model for spatially correlated errors is a model in which independent noise processes acts on overlapping regions, where the overall error is given by the product of these local errors. An example of this is the 
%% PURPOSEFUL AMERICAN
`nearest-neighbor depolarizing' model of Ref.~\cite{NN-depolarising}, in which each nearest-neighbour pair is afflicted uniformly by any non-trivial 2 qubit error with probability $1-p$. We discuss modifying the stat mech mapping for models such as this in \cref{app:overlapping}.

% Section 2: Correlated noise

\section{Correlated noise}
\label{sec:correlated}

Now that we have reviewed the stat mech construction for \emph{independent} Pauli noise, we now want to consider extending this to \emph{correlated} noise models. 

We will consider a model with {local spatial correlations}, which forms a natural and systematic way of relaxing the independence condition above. Specifically, we want to consider models in which errors at sufficiently distant locations are conditionally independent, but in which they may be arbitrarily correlated at short range. The conditional dependences described above are naturally represented through \emph{probabilistic graphical models}, such as Markov random fields and Bayesian networks. Below we will focus on a simpler and more mathematically convenient model of \emph{factored distributions}. We discuss the relationship between these models in more detail in \cref{app:graphical}.

\subsection{Factored distributions}

For a global error $E\in\mathcal P^n$, let $E_i$ denote the action on a single site $i$. Similarly, for some set of sites \mbox{$R=\lbrace i_1,i_2,\dots\rbrace$}, let $E_R$ denote action of $E$ on sites in $R$.

For the previous mapping, we leveraged the fact that independence of two random variables $A$ and $B$ implies their joint distribution factors, $\Pr(A,B)=\Pr(A)\cdot\Pr(B)$. We now want to consider more general distributions which can also be locally factored.

\begin{definition}[Factored distribution]
\label{def:factored}
	A distribution \emph{factors} over sets $\lbrace R_j\rbrace_j$ if there exists non-negative functions \mbox{$\lbrace \phi_j:\mathcal{P}_{R_j}\to \mathbb{R}^{+}\rbrace_j$} such that
	\begin{align}
		\Pr(E)=\prod_j\phi_j(E_{R_{j}}).
	\end{align}
\end{definition}

If the regions $\lbrace R_j\rbrace_j$ are disjoint, then these just correspond to distributions which are independent after appropriate coarse-graining. 
We, however, \emph{do not} require that these regions are disjoint, which allows us to consider genuinely correlated noise models, which remain correlated even after coarse-graining.
In another extreme limit, if there is only a single region $R$ and it contains \emph{all} of the random variables, then we can take $\phi_R$ to be simply the complete joint probability distribution. 
This shows that the structure of the individual regions $R_j$ can interpolate between the case of independent noise (when $R_j = \{j\}$) and a general distribution (when $R = \{1,\ldots,n\}$).
Intermediate cases correspond to probability distributions with differing ranges of correlation. 
Thus the factored distribution formalism forms a natural ansatz for describing finite-range correlations efficiently in a probabilistic model. 
We refer the reader to \cref{app:graphical} for more discussion of factored distributions.

\subsection{Correlated statistical mechanical mapping}

We now extend the stat mech mapping to the case of correlated noise. 
Specifically we will consider an error model which factors over $\lbrace R_j\rbrace_j$, with factors $\lbrace\phi_j\rbrace_j$. 
In the mapping for independence noise, each interaction term corresponding to a single-site Pauli. 
As hinted at earlier, we can account for correlated noise models by including interaction terms corresponding to multi-site $\sigma$, specifically $\sigma$ living on a single region $R_j$. 
Recall that $S_k$ are the stabilisers of the quantum code and $c_k$ are the associated $d$-level classical degrees of freedom taking values in $\mathbb{Z}_d$.
Then the stat mech Hamiltonian takes the following form.

\begin{definition}[Stat mech Hamiltonian: \emph{correlated noise}]
	\label{defn:ham2}
	For a Pauli $E\in \mathcal{P}^{n}$, and coupling strengths \mbox{$\lbrace J_j:\mathcal{P}_{R_j}\to \mathbb{R} \rbrace_j$}, the stat mech Hamiltonian $H_E$ is defined as
	\begin{align}
	\label{eq:StatMechHam2}
	H_E(\vec c)
	=-\sum_{j,\sigma\in \mathcal{P}_{R_j}}J_j(\sigma)\comm{\sigma}{E}\prod_k\comm{\sigma}{S_k}^{c_k}\!.
	\end{align}
\end{definition}
We notice that we can once again use the multiplicativity of the scalar commutator to rewrite this in the more mathematically convenient form
\begin{align}
H_E(\vec c)
&=-\sum_{j,\sigma\in \mathcal{P}_{R_j}}\!\!J_j(\sigma)\comm{\sigma}{E\prod_kS_k^{c_k}},
\end{align}
and see that our Hamiltonian retains the symmetry we leveraged in the independent case,
\begin{align}
	H_{ES_k}(\vec{c})=H_E(\vec{c}+\hat{k}).
\end{align}

In order for the definition of the stat mech Hamiltonian to connect to the error probabilities of the quantum code for a given noise model, we need to find an analogue of the Nishimori conditions.
Fortunately, our formulation of this condition in the independent case gives an immediate generalisation.

\begin{definition}[Nishimori conditions: \emph{correlated noise}] 
	\label{defn:nishimori2}
	An inverse temperature $\beta$ and coupling strengths $\lbrace J_j\rbrace_j$ satisfy the \emph{Nishimori conditions} with respect to factors $\lbrace \phi_j\rbrace_j$ if
	\begin{align}
	\label{eq:CorrelatedNishimori}
	\beta J_j(\sigma)
	=
	\frac{1}{|\mathcal{P}_{R_j}|}
	\sum_{\tau\in\mathcal P_{R_j}} \log\phi_j(\tau)\comm{\sigma}{\tau^{-1}},
	\end{align}
\end{definition}
Using a proof analogous to the independent case \cref{thm:fundie}, we also get a \emph{fundamental theorem} for this correlated model.
\begin{theorem}[Fundamental theorem of stat mech mapping: \emph{correlated noise}]
	\label{thm:fundie2}
	Given the stat mech Hamiltonian \emph{(}\cref{eq:StatMechHam2}\emph{)} associated to a noise model that satisfies the Nishimori condition \emph{(}\cref{eq:CorrelatedNishimori}\emph{)}, the probability of a logical class of errors is equal to the corresponding partition function,
	\begin{align}
	Z_E=\Pr(\bar E).
	\end{align}
\end{theorem}

As with the independent case, we will see that this similarly implies that the error-correction threshold manifests as a quenched phase transition (\cref{sec:decoding}), and that this allows us to construct efficient tensor network approximations to the maximum likelihood decoder (\cref{sec:mld}).

\subsection{Noise Hamiltonian}

Consider the case where each factor $\phi_j$ is strictly positive. Suppose we define the local Hamiltonian
\begin{align}
	\tilde H(E):=-\sum_{j}\log\phi_j(E),\label{eqn:noiseham}
\end{align}
which we refer to as the \emph{noise Hamiltonian}. The states of this Hamiltonian are labelled by Pauli errors, elements of $\mathcal{P}^n$. The noise model then corresponds to thermal distribution of this Hamiltonian at inverse-temperature $\beta=1$,
\begin{align}
	\Pr (E) = e^{-\tilde H(E)}.
\end{align}

Expressed in this way, the Nishimori condition, \cref{defn:nishimori2}, can be seen as a relationship between the stat mech Hamiltonian of \cref{defn:ham2} and the noise Hamiltonian of \cref{eqn:noiseham}. Specifically, it takes the form
\begin{align}
	\beta H_E(\vec c)=\tilde{H}\left(E\prod_{k}S_k^{c_k}\right).
\end{align}
In this sense we see that, when the Nishimori condition is satisfied, the interactions within the stat mech Hamiltonian $H_E$ naturally correspond to those in the noise Hamiltonian $\tilde H$, but with a change in the underlying state space.

\subsection{Example: Toric code with correlated bit flips}
\label{subsec:correlated bitflip}

\begin{figure}
	\centering
	\includegraphics{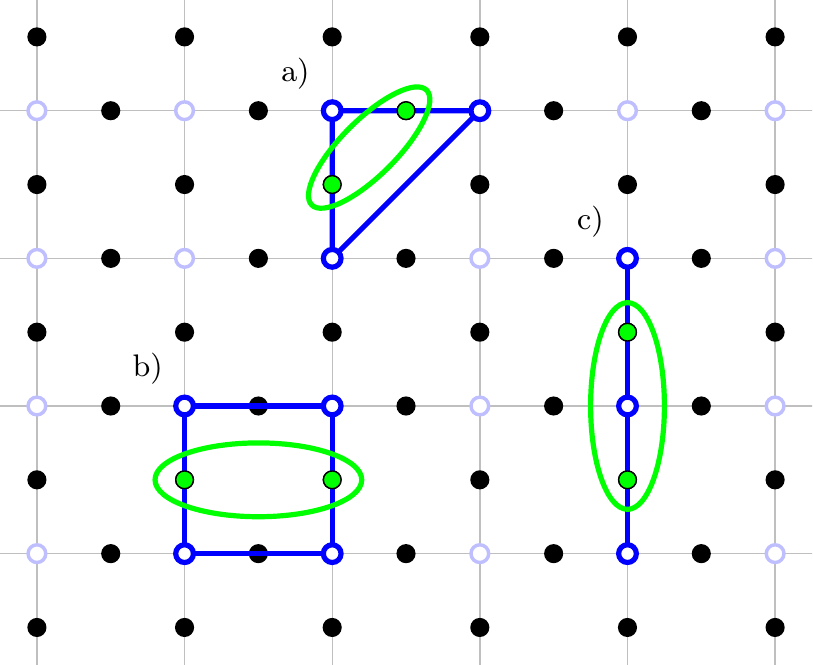}
	\caption{
		Stat mech couplings induced by pairwise correlated bit-flips in the toric code. Solid circles indicate qubits, and hollow circles the spins of our stat mech model, corresponding to stabiliser generators. Green ellipses denote the type of correlations, specifically the regions over which the error model factors. These correspond to the interactions in the noise Hamiltonian \cref{eqn:noisehamtoric}. Blue edges connecting spins indicate the couplings induced in the statistical mechanical model by such correlations. The labelled couplings a), b) and c) respectively correspond to:	a) nearest-neighbour correlations induce a two-body diagonal coupling, b) next-nearest neighbour (across-plaquette) correlations induce a four-body face coupling, c) next-nearest-neighbour (across-vertex) induce two-body distance-2 couplings.
	}
	\label{fig:nntoric}
\end{figure}

We now consider an example of a correlated noise model: correlated bit-flips in the toric code. Consider a noise model defined by an Ising noise Hamiltonian with coupling $\tilde J$ and field strength $\tilde h$,
\begin{align}
	\tilde{H}%\left(\otimes_e X^{\frac{x_e+1}{2}}\right)
	=- \sum_e \tilde h\,x_e-\sum_{e\sim e'}  \tilde J\,x_ex_{e'},
	\label{eqn:noisehamtoric}
\end{align}
where we have chosen the convention that $x_e=-1$ corresponds to $E_e=X$ and $x_e=+1$ to $E_e=I$. Here $\tilde J$ controls the magnitude (and sign) of the correlations, with $\tilde J=0$ corresponding to independent errors, $\tilde J>0$ to bunching errors, and $\tilde J<0$ to anti-bunching errors.

For a given error $E$ (and the corresponding values of the spin variables $x_e$), the stat mech model is of the form
\begin{align}
	H_E = -\sum_e \left(h_ex_e\right){\color{blue}\prod_{v\in\partial e}s_v} - \sum_{e\sim e'}\left(Jx_ex_{e'}\right){\color{blue}\prod_{v\in\partial(ee')}s_v},
\end{align}
where $\partial(ee')$ denotes the vertices that surround either $e$ or $e'$, but not both. Here the variables $\color{blue} s_v$ form the degrees of freedom, and the $x_e$ form the quenched disorder. The Nishimori conditions for this model reduce to \mbox{$\beta J=\tilde J$} and \mbox{$\beta h=\tilde h$}. 

As in the independent case, we can see that the $\tilde h$ field term in our noise Hamiltonian has induced a 2D random-bond Ising model. The addition of the $\tilde J$ term has induced additional longer range couplings. 
The geometry of these couplings is shown is \cref{fig:nntoric} for the case where the noise correlations couple nearest-neighbour and next-nearest-neighbour qubits. 
The corresponding stat mech Hamiltonian has a similar locality, and contains at most 4-body interactions among the ${\color{blue}s_v}$ degrees of freedom.

% Section 3: Numerics

\section{Numerics}
\label{sec:numerics}

A key advantage of the stat mech mapping is that it allows us to reappropriate techniques for determining the phase diagrams of classical spin systems for approximating the thresholds of quantum codes. 
By way of example, we consider using Monte Carlo simulations to determine the threshold of the toric code under a correlated model of bit-flip noise.

We shall consider the model of `across-plaquette' correlated bit-flips described in \cref{subsec:correlated bitflip} (see \cref{fig:nntoric}b). 
We will restrict our attention to noise which obeys certain natural symmetries, namely, that the correlations are site-independent and symmetric between correlated errors. 
This is equivalent to saying that the noise Hamiltonian has symmetric, site-independent interaction terms. 

Instead of expressing this model in terms of an Ising noise Hamiltonian, it will be convenient to parameterise this model in terms of the marginal error rate. 
Let $p$ denote the marginal error probability on any site, i.e.\ $\Pr(E_e=X)=p$ on any edge $e$. 
Suppose that the error probability given that a neighbouring error has or has not occurred is $p_\pm$, so that 
\begin{subequations}
\begin{align}
\Pr(E_e=X|E_{e'}=X)=p_+,\\
\Pr(E_e=X|\mathrlap{E_{e'}=\,I}\phantom{E_{e'}=X})=p_-,
\end{align} 
\end{subequations}
where $e$ and $e'$ lie on opposite sides of a plaquette. 
For the marginal probability to be $p$, these probabilities are subject to the consistency condition $p=pp_++(1-p)p_-$. 

If $p_-<p<p_+$ the errors tend to bunch together, whereas they tend to anti-bunch if $p_->p>p_+$. 
A natural way of parameterising these correlations is in terms of the \emph{correlation parameter}, 
\begin{align}
	\eta:=p_+/p_-,
\end{align}
where $\eta=1$ corresponds to uncorrelated noise, $\eta>1$ to bunched, and $\eta<1$ to anti-bunched.  

In the limit of infinite bunching ($\eta\to \infty$) the model produces exclusively logical errors, meaning that the threshold vanishes, $p_{\mathrm t}\to 0\%$. 
Similarly, for infinite anti-bunching ($\eta\to 0$) the model cannot produce non-trivial logical errors, and we expect that $p_{\mathrm t}\to 50\%$~\cite{Tuckett2018}. 
We will be considering the case of mild bunching correlations, $\eta=2$, which one would expect to lower the threshold.

Applying the stat mech mapping, we get a random-bond Ising-type model on a square lattice, containing 2-body edge terms and 4-body face terms. The Hamiltonian takes the form
\begin{align}
H_E(\vec s)
=&-\sum_{e} \left(J_2 x_e\right) {\color{blue}\prod_{v\in \partial e}s_v}
%\\&\notag\qquad
-\sum_{e\sim e'} \left(J_4 x_ex_{e'}\right) {\color{blue}\prod_{v\in \partial(ee')}s_v},
\end{align}
where $x_e=+1$ if $E_e=I$ and $x_e=-1$ if $E_e=X$, and $e\sim e'$ denotes edges lying across a plaquette from each other. 

Normalising our Hamiltonian such that $J_2\equiv 1$, our system has two Nishimori conditions. When \mbox{$\eta=2$}, the first of these conditions is
\begin{align}
	J_4=\frac{\log \frac{1-p}{2}
		}{\log 4p^2}.
\end{align}
We will impose this condition, and determine the phase diagram of the system in the two remaining parameters: the error probability $p$ and temperature $T$. The threshold probability can be found by finding the intersection of this phase boundary with the second remaining Nishimori condition,
\begin{align}
\beta_{\textrm{Nish}}=-\log 4p^2.
\label{eqn:numericalnishimori}
\end{align} 

In \cref{fig:montecarlo} we show the phase diagrams of the \mbox{$\eta=1$} uncorrelated model, and our numerical results for the \mbox{$\eta=2$} correlated model. 
The details of our numerical simulation are found in \cref{app:simulation}. 
We find that that under these mild correlations the threshold drops to \mbox{$p_{\mathrm t}(\eta=2)=10.04(6)\%$}, from the uncorrelated threshold of \mbox{$p_{\textrm t}(\eta=1)=10.917(3)\%$}~\cite{ParisenToldin2009}, confirming the earlier intuition that correlated errors will indeed reduce the threshold. 

\begin{figure}
	\centering
	\includegraphics{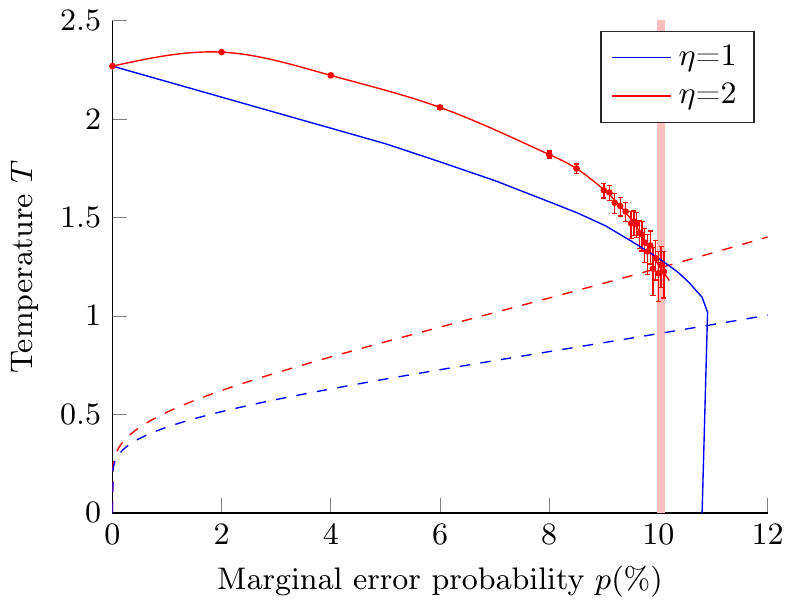}
	\caption{Phase boundary of the stat mech models corresponding to uncorrelated ($\eta=1$) and correlated ($\eta=2$) bit-flip noise in the toric code. The solid lines indicate the phase boundaries of the two models, with error bar indicating statistical uncertainty in our numerical results. The $\eta=1$ data is taken from Refs.~\cite{Honecker2001,ParisenToldin2009,Amoruso2004,Merz2002,Dennis2001}, and the $\eta=2$ comes from our numerics, as detailed in \cref{app:simulation}. The dashed lines indicate the corresponding Nishimori conditions. The shaded red region indicates our estimate of $p_{\mathrm t}=10.04(6)\%$ for the threshold of the correlated model.}
	\label{fig:montecarlo}
\end{figure}

% Section 4: Circuit model noise

\section{Spatio-temporal correlations}
\label{sec:circuit}

In \cref{sec:correlated} we considered the stat mech mapping for error models which factor. 
This model applies in the important case of spatially correlated errors, followed by ideal measurements. 
We would like to generalise this construction to the case of multiple rounds of syndrome measurements, subject to \emph{spatio-temporal} correlations in the noise. 
We will do this by constructing a subsystem code which, subject to purely spatially correlated noise, reproduces the measurement statistics of our original stabiliser code under a spatio-temporally correlated noise model. 
We refer to this construction as the \emph{history code}. 
Doing this in such a way that the locality of our code and noise model are preserved, we can then apply the construction of \cref{sec:correlated} to the history code to give a corresponding stat mech model.

As well as presenting the construction of the history code, we will also discuss an important family of spatio-temporally correlated noise models which possess a correlation structure which make them amenable to this construction: circuit-based noise. 
In this model we consider a syndrome measurement procedure composed of faulty gates. 
The key distinction between this and purely spatial correlations is that errors incurred at an earlier time can be spread around by subsequent measurement circuits. 
Considering Clifford measurement circuits subject to Pauli noise, this gives a spatio-temporally correlated Pauli noise model. 
By applying the construction of \cref{sec:correlated} to the corresponding history code, we get a stat mech model with a phase transition corresponding to the \emph{fault-tolerant threshold} of the original code.

\subsection{History code}
\label{subsec:history}

The idea behind the history code is to convert time into an additional spatial dimension, allowing us to naturally convert spatio-temporal correlations into purely spatial correlations, albeit living in one dimension higher. 
Specifically, for each site $i$ in our original code, and measurement round $t$, the history code will have a corresponding site $(i,t)$.

\subsubsection{Noise model}

The error models we shall consider are those which factor in a spatio-temporal sense. Specifically, let $E^{(t)}$ denote the error to which our code is subjected \emph{prior} to the $t$th round of measurements, and $E:=(E^{(1)},\dots,E^{(T)})$ denote the \emph{error history}. We will consider models in which the distribution of error histories factors, i.e.\ 
\begin{align}
\Pr(E) = \prod_l\phi_l(E),
\end{align}
where the supports of $\phi_l$ are local in both space \emph{and} time.

As the name may suggest, we will take our error model on the history code to be given by the error \emph{histories} of our noise with spatio-temporal correlations. 
Specifically let $\hat E^{(t)}$ denote the action of $E^{(t)}$ on the $t$th layer of the history code, and correspondingly let the action of an error history be denoted $\hat E:=\bigotimes \hat E^{(t)}$. Given this action, the distribution of errors is simply given by the original error model, and so
\begin{align}
\Pr(\hat E)=\prod_l \phi_l(\hat E).
\end{align}
As $\hat E$ can be interpreted as a purely spatial error in the history code, instead of an error history, this now constitutes a noise model with purely spatial correlations.

\subsubsection{Gauge generators}

Let $M_{j}^{(t)}$ denote the $j$th Pauli measurement occurring in the $t$th round of syndrome measurements. 
For error correction to remain well-defined, the final measurement round will need to consist of ideal measurements of all the stabiliser generators%
\footnote{This is required to preclude errors occurring after measurements, which are clearly uncorrectable. This assumption is equivalent to only requiring that an error correction procedure correct errors prior to the final measurement round for it to be deemed successful.}. 
As such if we let $T$ be the number of measurement rounds, this assumption means that $\left\lbrace M_j^{(T)}\right\rbrace_j=\lbrace S_k\rbrace_k$. 
We will require the measurements within each round have disjoint support with the exception of the special final round.

Let $\hat M_j^{(t)}$ denote the action of $M_j^{(t)}$ on the corresponding layer of the history code. By construction, we have that the measurement statistics of these operators reproduce those of the measurements in the original code, as
\begin{align}
	\comm{\hat M_j^{(t)}}{\hat E}=\comm{M_j^{(t)}}{E^{(t)}}
\end{align}
for any $j$ and $t$.

We now want to consider the history code itself, which takes the form of a subsystem code. To find the gauge generators, we can consider when two error histories are logically equivalent. As it is only the final accumulated error which determines whether or not a logical error has occurred, any two error histories are logically equivalent if and only if they possess the same syndrome, and do not differ by a non-trivial logical operator in the original code on the final time-slice. 
This implies that the gauge group is generated locally, and that these generators come in two forms: for $t<T$, the generators are given by the generators of the centraliser of $\hat M_{j}^{(t)}$ on the support of $\hat M_{j}^{(t)}$, and for $t=T$ they are just given by the measurement operators themselves $\lbrace \hat M_{j}^{(T)}\rbrace_j$.

\subsection{Circuit-based model}

We now turn our attention to a spatio-temporally correlated noise model which arises in the study of fault-tolerant error correction---circuit-based noise. In this noise model we explicitly take into account the circuits used to implement our syndrome measurements, and consider the faults within the constituent gates.

\subsubsection{Measurement circuits}

\begin{figure}
	\includegraphics{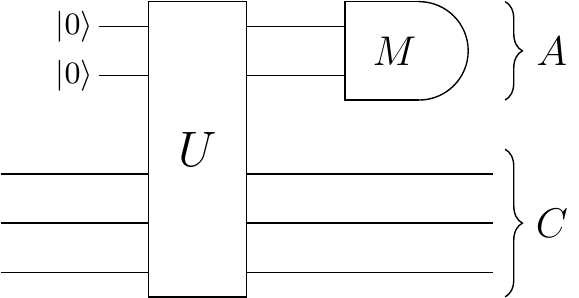}
	\caption{The form of measurement circuits we shall be considering. For each stabiliser generator $S$, several ancillae are prepared in the state $|0\rangle$, an entangling Clifford $U$ is performed, and then a Pauli measurement $M$ is performed on the ancillae. Sets $C$ and $A$ are used to denote the code and ancillary qubits involved in this measurement.}
	\label{fig:canonical}
\end{figure}

We will consider measurement circuits of the form shown in \cref{fig:canonical}: several ancillae are prepared in the state $|0\rangle$, a Clifford gate is applied to these and the qudits to be measured, and finally a Pauli measurement is performed on the ancillae. For the measurement circuit corresponding to $S_{j}^{(t)}$, we let $U_{j}^{(t)}$ denote the Clifford gate applied and $M_{j}^{(t)}$ the Pauli being measured. 
We denote the set of code qudits and ancillae involved in this syndrome measurement by
\begin{subequations}
\begin{align}
	C_j^{(t)}&:=\mathrm{supp}\left(S_j^{(t)}\right),
	\\
	A_j^{(t)}&:=\mathrm{supp}\left(M_j^{(t)}\right),
\end{align}
\end{subequations}
and let \mbox{$R_j^{(t)}:=C_j^{(t)} \cup A_j^{(t)}$} denote the full set of qudits involved, such that \mbox{$\mathrm{supp}\left(U_j^{(t)}\right) = R_j^{(t)}$}. The assumption that our measurements in each round are non-overlapping is equivalent $\left\lbrace C_{j}^{(t)}\right\rbrace_j$ being non-overlapping for each $t$. The layout of the history code corresponding to such a measurement procedure is shown in \cref{fig:circuit}.

As noted earlier, the final error correction round will correspond to an ideal measurement of all of the stabiliser generators. As such we will take $\lbrace M_j^{(T)}\rbrace_j=\lbrace S_k\rbrace_k$.

\subsubsection{Noise model}
\label{subsec:noisemodel}

\begin{figure}[h]
	\includegraphics{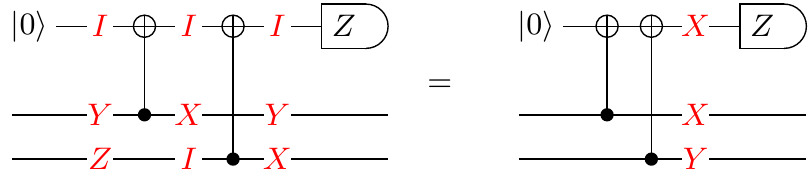}
%	\begin{align*}
%	\begin{aligned}
%	\Qcircuit @C=.4em @R=1em {
%		&&&&&\lstick{|0\rangle}&\push{\color{red}I}\qw&\targ&\push{\color{red}I}\qw&\targ&\push{\color{red}I}\qw&\meter\qw\\
%		\\\\
%		&\qw &\qw &\qw &\qw &\qw &\push{\color{red}Y}\qw &\ctrl{-3} &\push{\color{red}X}\qw &\qw &\push{\color{red}Y}\qw &\qw &\qw &\qw\\
%		&\qw &\qw &\qw &\qw &\qw &\push{\color{red}Z}\qw &\qw &\push{\color{red}I}\qw &\ctrl{-4} &\push{\color{red}X}\qw &\qw &\qw &\qw
%	}
%	\end{aligned}
%	\quad=\quad
%	\begin{aligned}
%	\Qcircuit @C=.4em @R=1em {
%		&&&&&\lstick{|0\rangle}&\push{\color{red}}\qw&\targ&\push{\color{red}}\qw&\targ&\push{\color{red}X}\qw&\meter\qw\\
%		\\\\
%		&\qw &\qw &\qw &\qw &\qw &\push{\color{red}}\qw &\ctrl{-3} &\push{\color{red}}\qw &\qw &\push{\color{red}X}\qw &\qw &\qw &\qw\\
%		&\qw &\qw &\qw &\qw &\qw &\push{\color{red}}\qw &\qw &\push{\color{red}}\qw &\ctrl{-4} &\push{\color{red}Y}\qw &\qw &\qw &\qw
%	}
%	\end{aligned}
%	\end{align*}
	\caption{General Pauli errors being pushed through a measurement circuit. Notice that the $Y$ error on the upper code qubit spreads onto the measurement qudit.}
	\label{fig:pullthrough}
\end{figure}

\begin{figure*}%[b!]
	\centering	
	\includegraphics{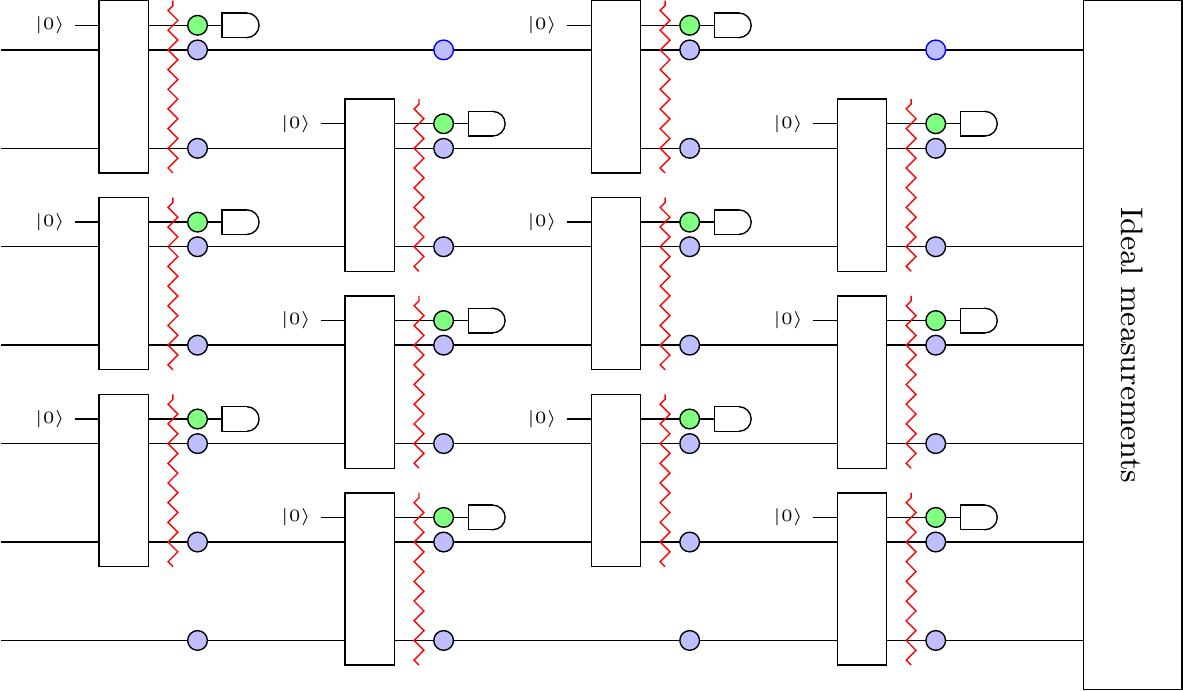}
	\caption{Laying out the qubits of the history code on a syndrome extraction circuit. The blue dots indicate code qubits contained in sets $\lbrace C_j^{(t)}\rbrace_{j,t}$, green dots indicate ancilla qubits contained in sets $\lbrace A_{j}^{(t)}\rbrace_{j,t}$, and red squiggles indicate the location of errors $\lbrace \epsilon_j^{(t)}\rbrace_{j,t}$.}
	\label{fig:circuit}
\end{figure*}

Consider now performing the above syndrome measurement circuits, subject to Pauli noise. For simplicity we will assume that the errors experienced within different measurement circuits are independent%
\footnote{More generally any model which produces factored spatio-temporally correlated noise could be used.}%
, but will allow arbitrary correlations within each circuit. 
Without loss of generality, we can push the error through the circuit, giving an effective error in $\mathcal P_{R_{j}^{(t)}}$ acting \emph{after} the application of $U_{j}^{(t)}$ (see \cref{fig:pullthrough}). 
Since we have assumed that our circuits are Clifford, this updated Pauli error can be computed efficiently. 
Let $p_{j}^{(t)}$ be the distribution on $\mathcal P_{R_{j}^{(t)}}$ of these pushed-through errors.

There are two error sources which contribute to $E^{(t)}$: errors which accrued from previous syndrome measurement rounds, and those which occurred in the $t$th round itself. 
This means that our errors satisfy the recurrence relation
\begin{align}
E^{(t)}=\epsilon^{(t)}\cdot U_tE^{(t-1)}U_t^\dag,\label{eqn:decomp}
\end{align}
where $\epsilon^{(t)}$ denotes the freshly introduced errors which occurred in the $t$th round, and $U^{(t)}$ to the unitary action of the syndrome measurement circuits in the $t$th round. 

Decomposing this further into the individual measurement circuits, we have
\begin{align}
U^{(t)}=\bigotimes_j U_{j}^{(t)}
\quad\text{and}\quad
\epsilon^{(t)}=\bigotimes_j \epsilon_{j}^{(t)},
\label{eqn:spatial}
\end{align}
where $\epsilon_{j}^{(t)}$ and $U_{j}^{(t)}$ are supported solely on $R_{j}^{(t)}$, and $\epsilon_{j}^{(t)}$ is distributed according to $p_{j}^{(t)}$. As our code is assumed to be error-free prior to the beginning of the syndrome measurements, we have the initial condition $E^{(0)}=I$, and so this recurrence relation entirely describes the error model.

\subsubsection{Factorising circuit noise}

We start by noticing that in \cref{eqn:decomp} the errors at time $t$ are entirely determined by the errors at time $t-1$, and the newly incurred errors $\epsilon^{(t)}$. 
As we have assumed that these new errors incurred in each round of syndrome measurements
are independent, our error model is therefore Markovian, allowing us to factorise our noise model in the temporal direction
\begin{align}
	\Pr(E^{(1)},\dots,E^{(T)})=\prod_t \Pr(E^{(t)}|E^{(t-1)}).\label{eqn:factor_temporal}
\end{align}

Next we can decompose this into the individual measurement circuits using \cref{eqn:spatial}. Specifically we know that the individual new errors $\epsilon_{j}^{(t)}$ are independently distributed according to $p_{j}^{(t)}$, which allows us to express the conditional probability as
\begin{align}
	\Pr\left( E^{(t)} \middle| E^{(t-1)} \right)
	=\prod_j p_{j}^{(t)}\left(\epsilon_{j}^{(t)}\right).\label{eqn:factor_spatial}
\end{align}

We now want to put this expression back purely in terms of the accumulated errors $E^{(t)}$, instead of the fresh errors $\epsilon^{(t)}$. Inverting the recurrence relation \cref{eqn:decomp}, we have that
\begin{align}
\epsilon^{(t)}=E^{(t)}\cdot U^{(t)}E^{(t-1)\dag} U^{(t)\dag}.
\end{align}
Recalling that $\epsilon^{(t)}$ and $U^{(t)}$ factor across $\lbrace R_{j,t}\rbrace_j$ (see \cref{eqn:spatial}), this implies that
\begin{align}
\epsilon^{(t)}_j=\left[E^{(t)}\right]_{R_j^{(t)}}\cdot U^{(t)}_j\left[E^{(t-1)\dag}\right]_{C_j^{(t)}} U^{(t)\dag}_j,
\end{align}
where $[P]_R$ denotes the restriction of a Pauli $P$ to a region $R$. Putting this together with \cref{eqn:factor_temporal,eqn:factor_spatial}, we have that the distribution on error histories \mbox{$\Pr(E^{(1)},\dots,E^{(T)})$} factors both in space and time, allowing us to apply the history code construction of \cref{subsec:history}.

% Section 5: Error correction as a disordered phase

\section{Error correction as a statistical mechanical phase}
\label{sec:decoding}

The fundamental theorem of the stat mech mapping, \cref{thm:fundie2}, links the equilibrium thermodynamic properties of our stat mech model and the error correction properties of our code via the disordered partition function and the error class probabilities. 
In fact, a much stronger connection is true, as has already been noted by previous authors in the case of independent noise, beginning with Dennis \textit{et al.}~\cite{Dennis2001}. 
Previous work has shown that the regions in parameter space for independent noise in which a code can and cannot be decoded are phases in the associated stat mech system~\cite{Dennis2001,Bombin2012,Kubica2017}. 
This implies that the error correction threshold of our code manifests as a phase transition in the associated stat mech model. 
Moreover, there is an explicit order parameter that captures this phase transition, although unfortunately it actually a ``disorder'' parameter in that it involves the amount of disorder in the stat mech model. 

In this section, we review and extend this connection to the general case of correlated noise.
We define appropriate notions of ``above'' and ``below'' the threshold and prove that the disorder parameter exhibits non-analytic behaviour in the sense that it converges or diverges if and only if the code is above or below the threshold, respectively. 
This shows a precise sense in which the phase boundary of the stat mech model is exactly the threshold of the corresponding code. 

In \cref{sec:numerics} we saw how, using this correspondence, numerical techniques to analyse phase transition in stat mech systems can be used to give approximations of code thresholds. We also note that this correspondence also opens the door to applying analytic techniques for studying phase transitions, such as the duality method~\cite{Bombin2012,Nishimori2007,Ohzeki2009}.

\subsection{Maximum probability and maximum likelihood decoding}

A decoder is an algorithm which takes as input the syndrome and attempts to estimate which error occurred. 
A natural starting point would be a decoder which simply outputs the most likely error in this error model among those which are consistent with the observed syndrome. 
We will refer to this as the \emph{maximum probability} (MP) decoder.

For degenerate codes, where several logically equivalent errors can have the same syndrome, the MP decoder is generally sub-optimal. 
This stems from the fact that successful decoding does not require the decoder to output precisely the error which occurred, but just an error which is logically equivalent. 
As such, the ideal decoding will not correspond to the single error with the highest probability, but the \emph{error class} with the highest probability. 
The optimal decoder therefore outputs an error from this most likely class, and we will refer to this as the \emph{maximum likelihood} (ML) decoder.

To see the sub-optimality of the maximum probability decoder, consider the following example. Suppose a syndrome $s$ is measured which is consistent with three errors, $E_1,E_2,E_3$, which occur with probabilities
\begin{align}
\Pr(E_1)=4\%,\quad
\Pr(E_2)=3\%,\quad
\Pr(E_3)=3\%.
\end{align}
If $E_2$ and $E_3$ are logically equivalent, then the error \emph{class} $\overline{E_2}$ is more likely than $\overline{E_1}$, even though $E_1$ itself is the single most likely error. In this case we can see that the MP decoder will be sub-optimal, and the conditional success probabilities for the two decoders are
\begin{subequations}
	\begin{align}
	\Pr\left( \text{MPD success} \middle| s \right)=40\%,\\
	\Pr\left( \text{MLD success} \middle| s \right)=60\%.
	\end{align} 
\end{subequations}

For each syndrome $s$, let $C_s$ denote an arbitrary Pauli with syndrome $s$. Similarly, for each logical $l$, let $L_l$ denote a Pauli corresponding to logical $l$. The sets $\lbrace \overline{C_sL_l}\rbrace_{s,l}$ correspond to the logical error classes, and form a partition of the Paulis $\mathcal P^n$. 

For a decoder to always return the code back to the code space, we require that when a syndrome $s$ is input, it always returns an error with this same syndrome $s$. As such, without loss of generality, a decoder can be taken to be of the form $s\mapsto C_sL_{\delta(s)}$, where $\delta$ is a map from syndromes to logicals. In this notation, the maximum probability decoder takes the form
\begin{align}
	\delta_{\textrm {MP}}(s):=\argmax_l \max_{E\in \overline{C_sL_l}}\Pr(E),
\end{align}
and the maximum likelihood decoder the form
\begin{align}
	\delta_{\textrm {ML}}(s):=\mathop\mathrm{arg\,max}_{l} \Pr(\overline{C_sL_l}).
\end{align}

\subsection{Minimum free energy and minimum energy decoding}

Let $H_E$ denote the stat mech model corresponding to our code and error model (see \cref{sec:mapping,sec:correlated,sec:circuit}), which satisfies the Nishimori conditions at inverse temperature $\beta_{\textrm{N}}$.

Similar to the maximum probability and likelihood decoders defined above, consider the $\beta$-\emph{minimum free energy} ($\beta$-MFE) decoder, which is given by minimising the free energy at inverse temperature $\beta$, 
\begin{align}
	\delta_{\beta\textrm{-MFE}}(s):=\argmin_l F_{C_sL_l}(\beta),
\end{align}
where $F_E(\beta):=-\frac{1}{\beta}\ln Z_E(\beta)$. 

When satisfying the Nishimori conditions, $\beta=\beta_{\textrm{N}}$, we can apply the fundamental theorem of the stat mech mapping (\cref{thm:fundie2}), to give
\begin{align}
	F_E(\beta_{\textrm{N}})=-\frac{1}{\beta_{\textrm{N}}}\ln \Pr(\overline{E}).
\end{align}
As this is a monotonically decreasing function of the error class probability, this tells us that the $\beta_{\textrm{N}}$-MFE decoder is precisely the maximum likelihood decoder,
\begin{align}
	\delta_{\beta_{\textrm{N}}\textrm{-MFE}}\equiv \delta_{\textrm {ML}}.
\end{align}

Similarly, if we take $\beta\to \infty$, then the free energy reduces to the minimum energy,
\begin{align}
	\lim\limits_{\beta\to\infty} F_E(\beta) = \min_{\vec s} H_E(\vec s). 
\end{align}
\cref{thm:fundie2} gives us that the energies of our stat mech model correspond to error probabilities, and so
\begin{align}
	\lim\limits_{\beta\to\infty} F_E(\beta) = \min_{E'\in \overline{E}}-\frac{1}{\beta_{\textrm{N}}}\ln\Pr(E').
\end{align}
As this is a monotonically decreasing function of the error probabilities, this tells us that the $\infty$-MFE decoder---which one could refer to as the \emph{minimum energy} decoder---is precisely the MP decoder,
\begin{align}
	\delta_{\infty\textrm{-MFE}}\equiv \delta_{\textrm {MP}}.
\end{align}

This reduction to MFE decoding implies that any of the plentiful methods for approximating partition functions~\cite{Metropolis1953,Hastings1970}, or free energy differences~\cite{Bennett1976} can be used to implement approximate ML decoding. 
In \cref{sec:mld} we will expand upon this connection, giving a tensor network algorithm which approximates ML decoding.

\subsection{The error correction threshold as a phase transition}

An important way of quantifying the resilience of a quantum error correction procedure to an error model is the \emph{error threshold}. Specifically, consider a family of quantum codes, with a logical algebra of finite dimension $K$, and an error model which depends on a parameter $\theta\geq 0$. 
We define the notion of threshold as follows. 
A code family has a \emph{threshold} if there exists a $\theta_t>0$ such that the asymptotic success probability is maximal for $\theta<\theta_t$, and minimal for $\theta>\theta_t$, i.e.\
\begin{align}
	\lim\limits_{n\to\infty} \Pr(\text{Decoder success})=\begin{dcases}
	1 & \text{if }\theta<\theta_t,\\
	1/K & \text{if }\theta>\theta_t.
	\end{dcases}
\end{align}
These regimes we refer to as being \emph{below threshold} and \emph{above threshold} respectively. 

Clearly this is a rather strong notion of threshold. 
One can imagine codes or noise models which possess an intermediate between being below or above threshold, where the code can be decoded better than random chance, but not perfectly.
As an example, having asymmetric rates of independent $X$ and $Z$ noise in the surface code will generate such a gap. 
Therefore, not all codes necessarily possess a threshold in this sense. 
However, the presence of a threshold in our sense is often taken as a desirable property in engineered quantum codes. 
We also note that for this section we have restricted our attention to codes with finite $K$.
Some codes with growing $K$, such as finite-rate LDPCs, are also known exhibit such an intermediate regime~\cite{Kovalev2018,Jiang2018}.

We also use the notion of threshold with respect to a particular choice of decoder, not just the optimal decoder. 
However, unless otherwise specified, we will be default to considering the threshold with respect to the optimal decoder, and refer to this as \emph{the} threshold.

Given that the threshold corresponds to a dramatic jump in the success probability of decoding, one might naturally suspect that this is precipitated by a corresponding dramatic jump in the error class probabilities. 
As these probabilities equal the partition function of our stat mech models, this would also suggest a (disordered) phase transition. 
To see that this is indeed the case, we will consider the disorder parameter given by the free energy cost of a non-trivial logical operator $L_m$,
\begin{subequations}
\begin{align}
\Delta_m(E)
&:=F_{EL_m}(\beta_{\textrm N})-F_{E}(\beta_{\textrm N})\\
&=-\frac{1}{\beta_{\textrm N}}\log Z_{EL_m}+\frac{1}{\beta_{\textrm N}}\log Z_E.
\end{align}
\end{subequations} 
We note that for topological codes which have string-like (or sheet-like in higher dimensions) logical operators, this corresponds to the free energy cost of a domain wall, as noted by Dennis \textit{et al.}~\cite{Dennis2001}. 
Let $\theta_p$ denote the point at which the stat mech model undergoes a phase transition corresponding to this parameter. 

We now want to show that the threshold corresponds to a phase transition in our stat mech model, \mbox{$\theta_t=\theta_p$}. 
As a first step, we consider the quenched average, \mbox{$\Delta_m:=\langle \Delta_m(E)\rangle_E$}. 
We would expect that $\Delta_m\to \infty$ below threshold and $\Delta_m\to 0$ above. 
The former limit was shown in Ref.~\cite{Kubica2017}, which implies that $\theta_p>\theta_t$. 
We present a simplified proof below.

\begin{lemma}[Divergence in mean]
	\label{lem:trans=thresh_mean}
	If the code is below threshold $\theta<\theta_t$, then the quenched average free energy cost diverges,
	\begin{align}
	\lim\limits_{n\to \infty}\Delta_m=\infty.
	\end{align}
\end{lemma}

\begin{proof}
	By recalling that $Z_E=\Pr(\overline E)$, we can see that $\Delta_m(E)$ only depends on $E$ up to logical equivalence. 
	A convenient set of representatives from each error class are given by the Paulis $\lbrace D_sL_l\rbrace_{s,l}$, where $s$ and $l$ correspond to the syndrome and logical degrees of freedom, and $\lbrace D_s\rbrace_s$ are the decoding Paulis associated with the MLD. 
	As shorthand, we will denote $\Pr(s,l):=\Pr(\overline{D_sL_l})$. 
	
	Using this, we can see that $\Delta_m$ can be expanded as a Kullback-Leibler divergence, 
	\begin{align}
	\Delta_m
	&=\frac{1}{\beta_{\textrm N}}\sum_{s,l}\Pr(s,l)\log \frac{\Pr(s,l)}{\Pr(s,l+m)}.
	\end{align}
	
	Applying the log sum inequality to the summation over syndromes, we get a bound in terms of the KL divergence of the marginal $\Pr(l)$,
	\begin{align}
	\Delta_m\geq \frac{1}{\beta_{\textrm N}}\sum_l\Pr(l)\log \frac{\Pr(l)}{\Pr(l+m)}.
	\end{align}
	We note that $\Pr(l=0)$ is the probability of the MLD succeeding. 
	Since we are below threshold we have that $\Pr(l=0)\to 1$, and so $\Delta_m\to\infty$.  
\end{proof}

Above the threshold one cannot simply consider the quenched average $\Delta_m$. 
To see this, consider a code in which a single unlikely syndrome can be perfectly decoded with certainty, and all others provide a uniform logical error. 
The correctable syndrome gives that $\Delta_m=\infty$, despite the fact that the code is above threshold. 
As such, we cannot necessarily conclude that $\Delta_m\to 0$ above threshold. 

Therefore, to provide a converse statement we must consider more than just the quenched average. 
Here we will formalise an argument first sketched in Ref.~\cite{Bombin2012} that shows that the way to obtain a converse is to change our notion of convergence. 

Suppose for the moment that $\Delta_m(E)$ concentrates around a value $\tilde\Delta_m$, in the sense that 
\begin{align}
\lim\limits_{n\to\infty}\Pr\Bigl(|\Delta_m(E)-\tilde\Delta_m|\leq \epsilon\Bigr)=1\quad \forall \epsilon>0.
\end{align}
A phase transition corresponds to a jump in this `typical' value $\tilde \Delta_m$. Because $\Delta_m(E)$ is not bounded from above, we notice that this typical value $\tilde \Delta_m$ need not necessarily correspond to the mean $\Delta_m$. 
Indeed, in the above counterexample, where $\Delta_m=\infty$ above threshold, we nonetheless see that $\tilde \Delta_m\to 0$ as expected. 
As such, we see that the mean is not necessarily the correct figure-of-merit. 
If instead we look to this typical value, then we can see that $\tilde\Delta_m\to\infty$ and $\tilde \Delta_m\to 0$ below \emph{and} above threshold respectively. 

\begin{lemma}[Divergence/convergence in probability~\cite{Bombin2012}]
	\label{lem:trans=thresh_prob}
	The code is below threshold if and only if the free energy cost of every non-trivial logical $m$ diverges in probability, 
	\begin{align}
	\lim\limits_{n\to\infty}
	\Pr\bigl( \Delta_m(E)\geq t \bigr)=1\quad\forall t,m\neq 0.
	\end{align}
	Similarly, the code is above threshold if and only if the free energy cost of every logical converges to zero in probability,
	\begin{align}
	\lim\limits_{n\to \infty} \Pr\bigl(\,|\Delta_m(E)|\leq \epsilon\,\bigr)=1\quad \forall \epsilon>0,m.
	\end{align}
\end{lemma}
\begin{proof}
	The success probability of the MLD corresponds to \mbox{$\Pr(l=0)$}. 
	We first note that $\Delta_m(E)$ can be written in terms of conditional probabilities,
	\begin{align}
	\Delta_m(E)=\frac{1}{\beta_{\textrm N}}\log \frac{\Pr(l|s)}{\Pr(l+m|s)}.
	\end{align}
	
	Next we recall that the maximum likelihood condition $\Pr(l=0|s)\geq \Pr(m|s)$ holds for all $s$ and $m$, and as such $\Pr(l=0)\geq 1/d^K$. 
	The code is below threshold if the probability of successful decoding approaches 1---this is equivalent to the probability of successful decoding \emph{conditioned upon the syndrome} approaching 1 almost surely, 
	\begin{align}
	\lim\limits_{n\to\infty}\Pr\Bigl[\Pr(l=0|s)\geq 1-\gamma \Bigr]=1\quad\forall \gamma>0.
	\end{align}
	To be clear, in this expression the inner probability is a probability over logicals $l$ and the outer probability is over syndromes. 
	Thus, the below threshold condition says informally that this distribution peaks at one in the limit. 
	This can be seen to be equivalent to being in the ordered phase,
	\begin{align}
		\lim\limits_{n\to\infty}
		\Pr\bigl( \Delta_m(E)\geq t \bigr)=1\quad\forall t,m\neq 0.
	\end{align}
	The forward and reverse directions are given by taking the choices of parameter
	\begin{subequations}
	\begin{align}
		t&:=\frac{1}{\beta_{\textrm N}}\log\frac{1-\gamma}\gamma,\\
		\gamma&:=1-\frac{1}{1+e^{-\beta_{\textrm N}t}(K-1)},
	\end{align}
	\end{subequations}
	respectively.
	
	In a similar vein, being above threshold implies that the conditional probability of decoding almost surely approaches $1/K$,
	\begin{align}
		\lim\limits_{n\to\infty} \Pr\Bigl[ \Pr(l=0|s) \leq 1/K+\eta \Bigr] = 1\quad\forall \eta>0.
	\end{align}
	This is equivalent to being in the disordered phase,
	\begin{align}
		\lim\limits_{n\to\infty}
		\Pr\bigl( |\Delta_m(E)|\leq \epsilon \bigr)=1\quad\forall \epsilon>0.
	\end{align}
	The forward and reverse directions here follow from the choices of parameter
	\begin{subequations}
	\begin{align}
		\epsilon:=\frac{1}{\beta_{\textrm{N}}}\log\frac{1+\eta K}{1-\eta K(K-1)},\\
		\eta:=\frac{1}{1+(K-1)e^{-\beta_{\textrm N}\epsilon}}-\frac{1}{K},
	\end{align}
	\end{subequations}
	respectively.
\end{proof}

We add more remark about the notions of above and below threshold that we have adopted here. 
As noted in the case of the surface code with independent $X$ and $Z$ noise with different strengths, there can be a gap between the cases of above and below threshold. 
In that case our theorem does not apply.
However, it is still the case that the steps of our proof could be followed to establish that there are a sequence of phase transitions in the disorder parameters $\Delta_m$ in the case where there are potentially differing thresholds for each logical subalgebras.

\subsection{Reentrance}

Above we have shown that, along the Nishimori line, there exists a phase transition which corresponds to the error correction threshold of the optimal decoder. 
By considering non-optimal decoders, we can extend this to show the equivalence between phase transitions and thresholds away from the Nishimori line.

By way of example, we consider the $\beta$-MFE decoders. 
We could now consider the phase diagram of our stat mech system in $(\theta,T)$-space, where $\theta$ was the parameter of our noise model, and $T=1/\beta$ is the temperature. 

\begin{lemma}[Reentrance of the ordered phase]
	For any $\beta$ such that the $\beta$-MFE decoder has a threshold, this threshold $\theta_t(\beta)$ is always less than or equal to that seen on the Nishimori line,
	\begin{align}
		\theta_t(\beta)\leq \theta_t(\beta_{\textrm N}).
	\end{align}
\end{lemma}
\begin{proof}
	As the $\beta$-MFE decoder is directly defined in terms of minimising the free energy, \cref{lem:trans=thresh_mean,lem:trans=thresh_prob} both naturally extend to show that the stat mech model has a phase transition at inverse temperature $\beta$ and $\theta=\theta_t(\beta)$. 
	
	As the ML decoder is optimal, it necessarily has the highest threshold of any decoder. Recalling that the ML decoder is precisely the $\beta_{\textrm N}$-MFE decoder, where $\beta_{\textrm N}$ is the inverse Nishimori temperature, this implies that \mbox{$\theta_t(\beta)\leq \theta_t(\beta_{\textrm N})$} for all $\beta$. 
\end{proof}

\begin{figure}
	\includegraphics{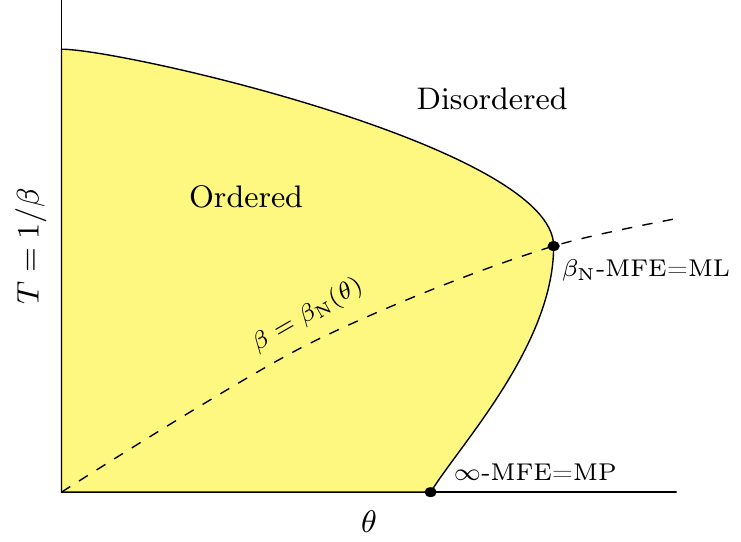}
	
	\caption{A sketch of the generic phase diagram for the stat mech model. The solid line indicates the phase boundary, and the shaded region the phase in which error correction is possible. The dashed line indicates the Nishimori condition, along which the MFE decoder corresponds to the ML decoder. The thresholds of the MP and ML decoder are both indicated by black dots. An explicit example, for the case of bit-flip noise in the toric code, is presented in \cref{fig:montecarlo}.}
	\label{fig:phase_sketch}
\end{figure}

% Section 6: Tensor network maximum likelihood decoding

\section{Tensor network maximum likelihood decoding}
\label{sec:mld}

In \cref{sec:decoding} we showed that the problem of maximum likelihood decoding can be reduced to calculating partition functions/free energy differences. Specifically it reduces to calculating 
\begin{align}
	\delta_{\textrm{ML}}(s):=\argmax_l  Z_{C_sL_l}(\beta_{\textrm N}),
\end{align}
where $\beta_{\textrm N}$ is the inverse Nishimori temperature. This relationship was first discussed in Ref.~\cite{Dennis2001}.

The problem of calculating the partition function of a stat mech model can be naturally expressed as the contraction of a tensor network~\cite{Verstraete2006,Bridgeman2017}, as we will review below.
This suggests a general method for approximating the ML decoder by finding a tractable (possibly approximate) contraction sequence for the tensor network associated to the partition function. 

The idea of a such a \emph{tensor network decoder} is not new. 
In Ref.~\cite{Bravyi2014}, Bravyi, Suchara and Vargo (BSV) consider an explicit tensor network for the ML decoder for the surface code with independent Pauli noise, and numerically study the threshold for some parameter choices. 
It is clear from their paper that this method generalises to other codes with independent Pauli noise models~\cite{Tuckett2017}, though the specific methods they use for (approximately) contracting the tensor network would need to be modified if the code were not a planar code. 
Other authors have also considered tensor network decoders which can account for non-Pauli noise~\cite{Darmawan2017,Darmawan2018}. 

In fact, the tensor network considered by BSV~\cite{Bravyi2014} is exactly the natural tensor network that one obtains from the partition function using the mapping of Refs.~\cite{Verstraete2006,Bridgeman2017}. 
Thus, the tensor network decoders that we define below are generalisations of the BSV decoder that, by virtue of \cref{thm:fundie2}, apply even in the case of correlated noise models on arbitrary codes. 
If the code has low-weight stabiliser (or gauge) generators, then storing these tensors is efficient. 
Contracting a general tensor network is unfortunately \#P-hard~\cite{AradLandau2008,SchuchWolfVerstraeteCirac2007}, however, we will discuss efficient approximate contraction strategies that could be employed in the interesting cases of spatially local codes.

\subsection{Tensor network algorithm for approximate maximum likelihood decoding}

To see how calculating the partition function of a stat mech model can be naturally expressed as the contraction of a tensor network, we begin with an arbitrary Hamiltonian
\begin{align}
	h(\vec s)=\sum_{j} h_j(\vec{s}),
\end{align}
with a corresponding partition function 
\begin{align}
	Z=\sum_{s} e^{-\beta h(\vec s)}=\sum_{s}\prod_{j} e^{-\beta h_j(\vec s)}.
\end{align}
This expression almost takes the form of a tensor contraction. To make this into a tensor network, consider indices $\lbrace \alpha_{i,j}\rbrace_{i,j}$, labelled by both spins and interactions. 
We want to require that $\alpha_{i,j}=s_i$ for all $j$, which we can impose by including Kronecker delta factors. 
In terms of these indices, the partition function can be expressed as
\begin{align}
	Z=\sum_{\alpha} 
	\left[
	\prod_i \delta\left(\lbrace \alpha_{i,j}\rbrace_j\right)
	%\right]\cdot \left[
	\prod_j e^{-\beta h_j\left(\lbrace \alpha_{i,j}\rbrace_i\right)}
	\right].
\end{align}
The virtue of this expression is that each $\lbrace \alpha_{i,j}\rbrace_{i,j}$ now occurs precisely twice in each term above, and as such the above expression forms a tensor network. 
Obviously, for local Hamiltonians in which $h_j$ only depends non-trivially on a finite set of sites $i$, we can drop any index $\alpha_{i,j}$ for which interaction $h_j$ is independent of $s_i$. 
After doing so, the tensor network simplifies to a network which has the same connectivity as the underlying Hamiltonian. 
Specifically, there is a tensor corresponding to each site ($\delta$) and each bond (Gibbs weight), with the natural connectivity. 
In \cref{fig:mldtoric} we show the corresponding diagram for the surface code under iid noise, which reduces to that found in Ref.~\cite{Bravyi2014}.

\begin{figure}[t]
	\centering
	\includegraphics{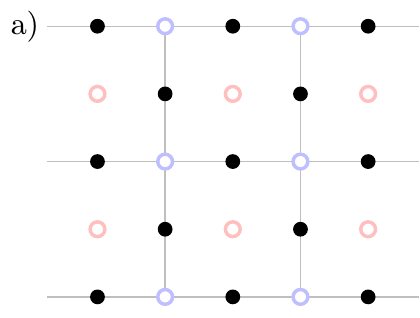}\includegraphics{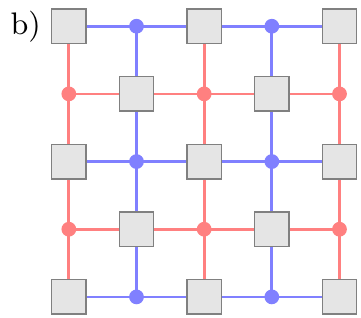}
	
	\includegraphics{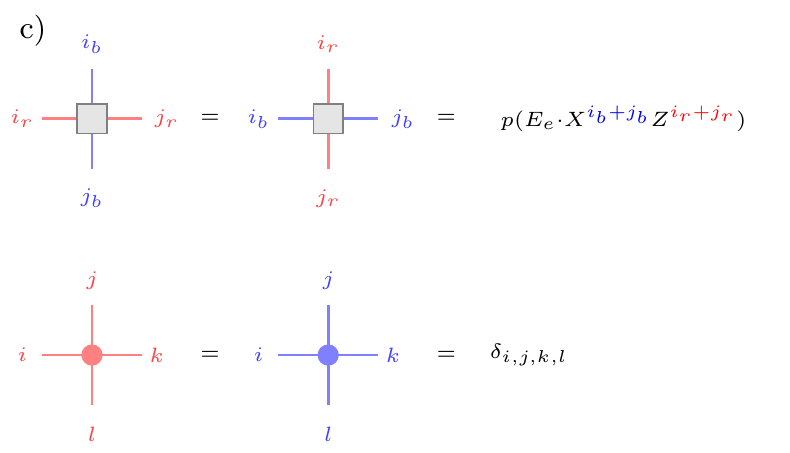}
	
	~\\[-.4cm]
	\caption{Tensor network ML decoder of the surface code under iid noise. a)~surface code and the associated statistical mechanical model. 
	Black circles indicate qubits, and the red/blue open circles the $Z$/$X$-type spins. 
	See \cref{fig:iidtoric} for the interactions within this model. 
	b)~Tensor network for the probability of an error class $\Pr(\overline E)$, or equivalently for the partition function $Z_E$, for iid noise distributed according to $p:\mathcal{P}\to\mathbb R$. 
	We note that this matches Fig.~6 of Ref.~\cite{Bravyi2014}. 
	c)~Values of the two tensors involved. 
	Red/blue dots are $\delta$ tensors, corresponding to the $Z$/$X$-type spins, and grey tensors correspond to Gibbs weights. 
	Once again this matches the tensors in Equations.~(39--41) of Ref.~\cite{Bravyi2014}.}
	\label{fig:mldtoric}
\end{figure}

Applying this construction to the stat mech mappings of \cref{sec:mapping,sec:correlated,sec:circuit}, we see that contracting this tensor network gives an algorithm to compute the ML decoding. 
Because contracting tensor networks is generally hard, we will also need \emph{approximate} tensor contraction schemes that provide parameterised approximations. 

In the case of 2D topological codes under spatially correlated noise, three approximate contraction schemes suggest themselves:
\begin{itemize}
	\item MPS-MPO contraction~\cite{Schollwock2011,VerstraeteCirac2004}. For the toric code, this method exactly reproduces the MPS decoder of Ref.~\cite{Bravyi2014} (see \cref{fig:mldtoric}).
	\item Transfer matrix~\cite{Onsager1944} and corner transfer matrix methods~\cite{Baxter1981, Baxter2007}. 
	\item Methods that involve renormalisation on the virtual level, e.g.\ tensor network renormalisation (TNR)~\cite{Evenbly2015,Bal2017} or (higher-order) tensor renormalisation group (HOTRG)~\cite{Levin2007, Xie2012}. These examples remain efficient when extended to more than two dimensions.
\end{itemize}
Each of these methods contain an approximation parameter in the form of the bond dimension. 
For a fixed bond dimension, these methods all provide polynomial time approximations to the ML decoder, though we do not know of any general results which control the approximation as a function of the bond dimension in any nontrivial way. 
Actually implementing these decoders (beyond what was done in Refs.~\cite{Bravyi2014, Tuckett2018}) or finding provable guarantees on the approximation to the ML decoder both remain open problems.

% Section 7: Conclusions
\section{Conclusions}
\label{sec:conclusions}

In this paper we have provided a broad extension of the stat mech mapping of Ref.~\cite{Dennis2001} to consider arbitrary stabiliser and subsystem codes, subject to correlated noise, including circuit-based noise. 
This class of noise models encompasses noise where distant spins are only conditionally independent, and allows for much more realistic noise modelling. 
As an application, we applied Monte Carlo simulation to this construction to show how positive correlations can push down the bit-flip threshold of the toric code. 
Finally, we showed how the stat mech mapping gives a natural family of efficient tensor network algorithms which approximate maximum likelihood decoding, generalising the decoder of Ref.~\cite{Bravyi2014}.

There are several natural avenues for further inquiry. 
For example, it is clear that our method should apply to the setting of continuous variables quantum codes such as GKP codes~\cite{Gottesman2000}, however there are analytical issues that must be addressed owing to the infinite dimensional Hilbert space. 
It would be interesting to understand these conditions in detail as these codes continue to gain experimental relevance~\cite{Ofek2016}. 

Even more interesting would be a formalism for deriving stat mech models that can handle non-Pauli errors such as amplitude damping or coherent errors, or for quantum codes that are outside the stabiliser code formalism such as commuting projector codes. 
Since the errors in the most general such models do not have an abelian action on the codewords, this raises the possibility that one would need a \emph{quantum} stat mech model to accurately capture the threshold in these cases. 
An interesting test case would be to extend our work to the semion code model of Ref.~\cite{Dauphinais2018}, since the anyons in the double semion model are abelian, but the check operators are not simple products of Paulis.

\section*{Acknowledgements}

We thank Michael Beverland, Jacob Bridgeman, Benjamin Brown, Marius Junge and Thomas Smith for helpful discussions. 
Special thanks go to Robin Harper for assistance with the numerics. 
The authors acknowledge the University of Sydney HPC service at The University of Sydney for providing HPC resources that have contributed to the research results reported within this paper. 
This work was initiated during the thematic semester on Analysis in Quantum Information Theory at the Institut Henri Poincar\'{e}, and the authors thank the Institut and the organisers for their hospitality during that time.

This research was supported by the Australian Research Council Centre of Excellence for Engineered Quantum Systems (Project ID CE170100009). 
CTC also acknowledges support Sydney Nano Postgraduate Scholarship (John Makepeace Bennett Gift).

\appendix

% Appendix A: Graphical models

\section{Graphical models}

\label{app:graphical}

An important family of correlated probability distributions is that of \emph{probabilistic graphical models}~\cite{Sucar2015}. In this appendix we briefly review the relationship between certain classes of graphical models and the factored distributions we consider in \cref{sec:correlated}. Specifically we will see that both undirected graphical models (Markov random fields) and directed acyclic graphical models (Bayesian networks) are both special examples of factored distributions. As such, the statistical mechanical mapping given in \cref{sec:correlated} also naturally extends to noise models of these forms as well.

\subsection{Markov random field}

We start by considering undirected graphical models, known as Markov random fields.

\begin{definition}[Markov random field]
	A distribution $\Pr(E)=\Pr(E_1,\dots,E_n)$ is a \emph{Markov random field} (MRF) with respect to an undirected graph $G$, if it satisfies the local Markov condition
	\begin{align}
	\Pr(E_i|E_{G\setminus i})=\Pr(E_i|E_{\partial i}),
	\end{align}
	for all $i$, where $\partial i$ denotes the neighbours of $i$ within $G$.
\end{definition}

For such a graphical model, Hammersley-Clifford theorem tells us that it can be expressed as a factored distribution.

\begin{lemma}[Hammersley-Clifford theorem]
	A strictly positive distribution which is a MRF with respect to $G$, can be factored over the set of all maximal cliques of $G$.
\end{lemma}
\begin{proof}[Proof sketch]
	For any subset $Y\subseteq G$, let
	\begin{align}
		\phi_Y(E_Y):=\prod_{X:X\subseteq Y}\Pr(E_X\otimes I_{G\setminus X})^{(-1)^{|Y\setminus X|}},
	\end{align}
	where $I$ was chosen as an arbitrary, but fixed, reference value of $E$.
	Noting that $\mu(X,Y)=(-1)^{|Y\setminus X|}$ forms a M\"obius function, we can apply M\"obius inversion theorem to give
	\begin{align}
		\Pr\left(E_Z\otimes I_{G\setminus Z}\right)=\prod_{Y\subseteq Z}\phi_Y(E_Y),
	\end{align}
	and more specifically
	\begin{align}
		\Pr\left(E\right)=\prod_{Y\subseteq G}\phi_Y(E_Y).
	\end{align}
	
	Applying the local Markov property, we find that \mbox{$\phi_Y\equiv 1$} for any $Y$ which is \emph{not} a clique.
	As such we find that
	\begin{align}
		\Pr(E)=\prod_j \phi_{C_j}(E_{C_j}),
	\end{align}
	where $\lbrace C_j\rbrace_j$ are the set of \emph{all} cliques in $G$. By grouping together factors, this can straightforwardly reduced to a factorisation over just \emph{maximal} cliques.
\end{proof}
\subsection{Bayesian Networks}
\label{app:bn}

We now consider directed cyclic graphical models, known as Bayesian Networks.

\begin{definition}[Bayesian network]
	A distribution $\Pr(E)=\Pr(E_1,\dots,E_n)$ is a \emph{Bayesian network} (BN) with respect to an directed acyclic graph $G$, if it satisfies the local Markov condition
	\begin{align}
		\Pr(E_i|E_{G\setminus \partial^-i})=\Pr(E_i|E_{\partial^+i}),
	\end{align}
	for all $i$, where $\partial^- i$ and $\partial^+i$ denote the descendants and parents of $i$ within $G$ respectively.
\end{definition}

As with Markov random fields, we see that Bayesian networks also naturally factorise.

\begin{lemma}[Factorisation of Bayesian Networks]
	A distribution $\Pr(\cdot)$ which is a Bayesian network with respect to a directed acyclic graph $G$, can be factored over $\lbrace \lbrace i\rbrace\cup \partial^+i \rbrace_i$. Specifically,
	\begin{align}
		\Pr(E)=\prod_{i=1}^n \Pr(E_i|E_{\partial ^+i}).
	\end{align}
\end{lemma}
\begin{proof}
	As our directed graph $G$ is acyclic, it possesses a topological order. Consider labelling our indices according to that ordering, such that $\partial^+i> i> \partial^-i$ element-wise. Next we consider expanding our joint distribution using the chain rule of conditional probabilities,
	\begin{align}
		\Pr(E_1,\dots,E_n)=\prod_{i=1}^{n} \Pr(E_i|E_{i+1},\dots,E_{n}).
	\end{align}
	Applying the local Markov condition, we get the desired factorisation,
	\begin{align}
	\Pr(E)=\prod_{i=1}^n \Pr(E_i|E_{\partial ^+i}).
	\end{align}
\end{proof}

% Appendix B: Numerics

\section{Numerical simulation details}
\label{app:simulation}

In this appendix we cover the details of how the statistical mechanical simulations of \cref{sec:numerics} were performed, specifically how the data presented in \cref{fig:montecarlo} was collected.

The simulations we present follow closely the techniques used in Ref.~\cite{Bombin2012} to study the toric code under a depolarising noise model.

\subsection{Order parameter}

As discussed in \cref{sec:numerics}, the system we are studying corresponds to a random-bond Ising model, with an addition four-body coupling. In the case of zero disorder, $p=0$, this system simply limits to the standard square-lattice Ising model. Noticing this, a convenient order parameter for this system is simply given by the average magnetisation,
\begin{align}
	m:=\frac{1}{L^2}\sum_is_i,
\end{align}
where $L$ is the linear size of our system, such that we have $L^2$ sites.

\subsection{Finite-size scaling}

\begin{figure}
	\centering
	
	% \fill[black!25] (axis cs:2.0539,-10) rectangle (axis cs:2.0651,10);
	% xlabel near ticks, ylabel near ticks, legend image post style={sharp plot,|-|}
	
	\includegraphics{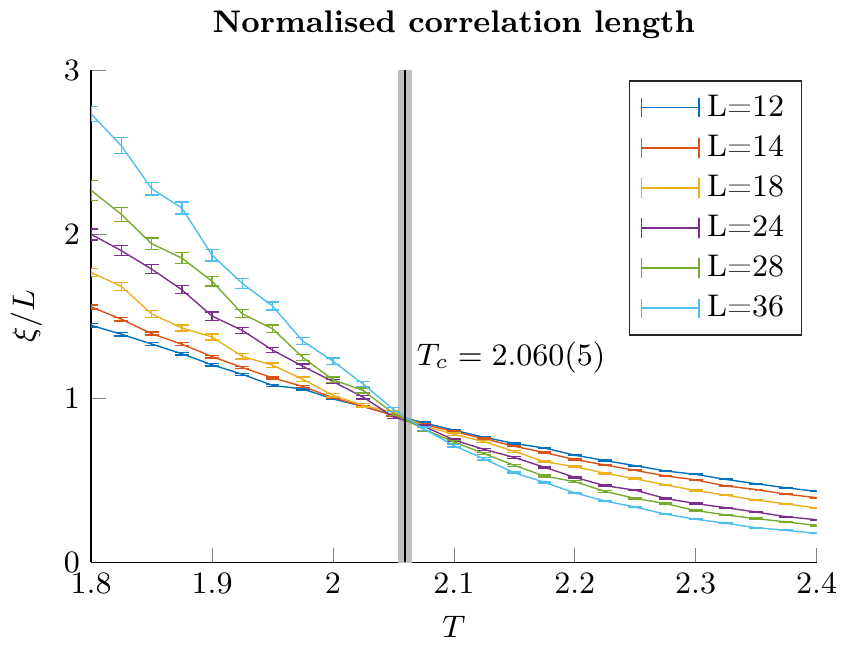}
	\includegraphics{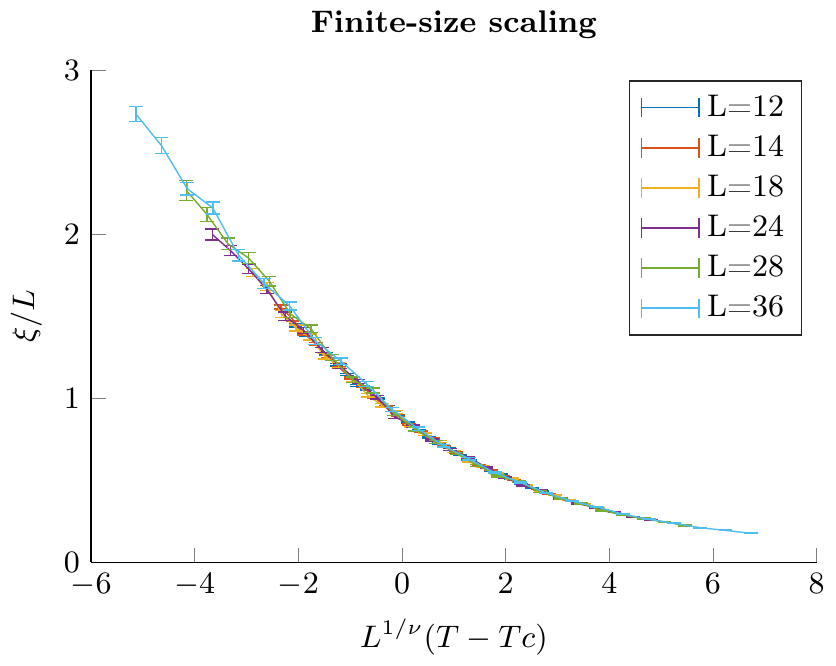}
	\caption{Crossing diagrams for \mbox{$p=6\%$}. a)~Normalised correlation length versus temperature, for several system sizes, with critical temperature indicated. b)~Finite-size scaling of normalised correlation length. All error bars indicated are 95\% confidence intervals given by bootstrapping.}
	\label{fig:crossing}
\end{figure}

For very large system sizes, we could sample the magnetisation directly, and find the phase transition by detecting when the system spontaneously magnetises. This na\"ive method would however require prohibitively large system sizes, especially for large disorder. Instead we will utilise this order parameter indirectly, by considering the finite-size scaling of the corresponding correlation length.

Following Ref.~\cite{Bombin2012}, we define the wave-vector-dependent magnetic susceptibility
\begin{align}
	\chi(\vec k):=\frac{1}{L^2}\left\langle \left|\sum_{i}s_ie^{i\vec{k}\cdot\vec{r_i}}\right|^2  \right\rangle,
\end{align} 
where $\vec{r_i}$ denotes the spatial position of site $i$, and $\langle\cdot\rangle$ denotes the thermal and disorder average,
\begin{align}
	\langle X\rangle := \sum_{E}\Pr(E)\cdot\sum_{s}\frac{e^{-\beta H_{E}(\vec s)}}{Z_E} X(\vec s).
\end{align}
Using this, we can now define the two-point finite-size correlation function
\begin{align}
	\xi:=\frac{1}{2\sin(k_{\min}/2)}\sqrt{\frac{\chi(\vec 0)}{\chi(\vec{k}_{\min})}-1},
\end{align}
where $\vec k_{\min}=(2\pi/L,0)$ is the minimal non-zero wave-vector. Near the phase transition this correlation length is expected to have the finite-size scaling~\cite{Palassini1999}
\begin{align}
	\label{eqn:fss}
	\xi/L\approx f\left[L^{1/\nu}(T-T_{\mathrm c})\right],
\end{align}
where $f$ is a dimensionless scaling function. 

At the critical temperature $T=T_{\mathrm c}$, the normalised correlation length $\xi/L$ becomes independent of temperature. We will determine the temperature of the phase transition by plotting $\xi/L$ as a function of $T$ for several different system sizes $L$, and fitting to \cref{eqn:fss}. If these curves do not cross, then this will be taken as indication that no phase transition is present.

\subsection{Simulation parameters}

\setlength\tabcolsep{.25cm}
\def\arraystretch{1.5}
\begin{table}[h!]
	\label{tbl:equil}
	\begin{tabular}{|c|c|ccc|}
		\hline
		$p$ (\%)     & $T$ & $\tau_{14}$ & $\tau_{24}$ & $\tau_{36}$ \\ \hhline{|=|=|===|}
		2.00         & 2.00       -- 2.50       & 14          & 16          & 18          \\
		4.00         & 2.00       -- 2.50       & 14          & 16          & 18          \\
		6.00         & 1.80       -- 2.40       & 14          & 16          & 18          \\
		8.00         & 1.60       -- 2.30       & 16          & 18          & 22          \\
		8.50         & 1.50       -- 2.10       & 17          & 21          & 23          \\
		9.00 -- 9.10 & 1.40       -- 2.00       & 18          & 22          & 24          \\
		9.20 -- 9.50 & 1.30       -- 2.00       & 18          & 22          & 25          \\
		\,~9.55 -- 10.20 & 1.25       -- 2.00       & 19          & 23          & 26          \\
		\,10.30 -- 10.50 & 1.25       -- 2.00       & 20          & 24          & 27          \\ \hline
	\end{tabular}
	\caption{Temperature ranges and equilibration times. An equilibration time of $\tau$ corresponds to $2^{\tau}$ Metropolis-Hastings steps (a number of single spin updates equal to the system size) before the correlation lengths recorded in the three last logarithmic bins are statistically consistent.}
\end{table}

We simulate our system for system sizes of \mbox{$L\in\lbrace 12,14,18,24,28,36\rbrace$}. Equilibration of our system is found by logarithmically binning the correlation length, and requiring that three successive bins agree to within error bars. To save time, the equilibration testing was only done on \mbox{$L=14,24,36$}, and these times were reused for \mbox{$L=12,18,28$}. The times and temperature ranges are given in \cref{tbl:equil}.

For $L=12,14$ the correlation length is recorded for 5000 disorder samples, 1000 samples are taken for $L=18,24$, and 500 for $L=28,36$. The temperature ranges are swept over with a resolution of $\Delta T=0.025$. Error bars are generated for all the data is given by the 95\% confidence intervals given by bootstrapping, using 10,000 resamples. An example crossing plot is given for $p=6\%$ in \cref{fig:crossing}. By performing this analysis for each $p$, we can estimate the phase boundary, as seen in \cref{fig:montecarlo}. All the simulations took a total of $1.1\times 10^6$ CPU-hours, or $120$ CPU-years.

\section{Overlapping independent errors}
\label{app:overlapping}

In this appendix we will consider a correlated noise model which, at least superficially, appears to have a similar structure to the factored noise considered in \cref{sec:correlated}. We will show that a statistical mechanical mapping for these noise models can be given by considering an appropriately enlarged code, and applying the mapping given for independent noise in \cref{sec:mapping}.

Consider a model which involves independent noise processes acting on overlapping regions, such that the overall noise is given by the product of these constituent errors. Specifically we have overlapping regions $\lbrace R_j\rbrace$, each subject to independent noise given by distributions $\lbrace p_j:\mathcal P_{R_j}\to\mathbb R^+\rbrace$. For each error $E_j\sim p_j$, the overall error is simply $E=\prod_j E_j$. When these regions are non-overlapping each $E_j$ simply corresponds to the restriction $\left.E\right|_{R_j}$, but for overlapping regions multiple constituent errors can give rise to the same overall error, e.g. both pairs of errors $\lbrace XYI,IIZ\rbrace$ and $\lbrace XII,IYZ\rbrace$ give the overall error $E=XYZ$. Due to this redundancy, these noise models take the more complicated form
\begin{align}
	\Pr(E) = \sum_{\substack{\lbrace E_j\rbrace:\,\prod_jE_j=E}} \prod_j p_j(E_j).
\end{align}

The approach here is similar to a stripped-down version of the history code construction considered in \cref{sec:circuit}. The idea here is to consider a code now acting on multiple copies of the original lattice, which we will refer to as the \emph{enlarged code}. This added redundancy will allow previously overlapping errors to act on disjoint copies, giving a truly independent noise model. This added redundancy will in turn be compensated for by introducing gauge generators between the layers.

\subsubsection{Enlarged code}

\begin{figure}
	\centering
	\includegraphics{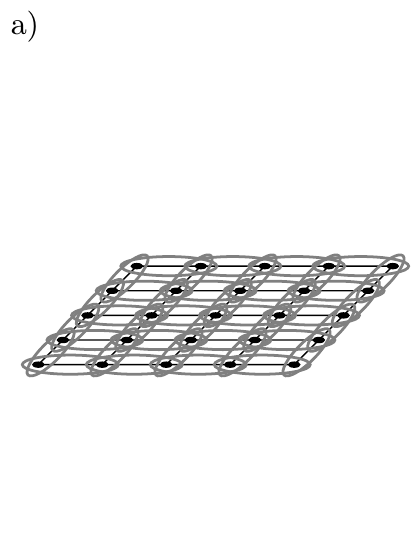}
	\includegraphics{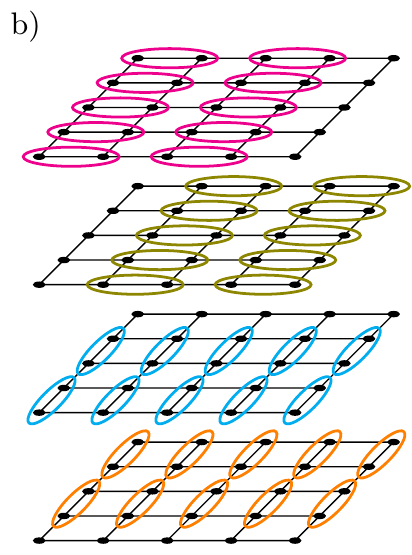}
	\caption{Overlapping-independent errors effecting nearest neighbour pairs on a square lattice, such as the 
		%% PURPOSEFUL AMERICAN
		`NN-depolarizing' model considered in Ref.~\cite{NN-depolarising}. a) Overlapping regions $R_j$ on the original code. b) Non-overlapping regions $\tilde R_j$ on $C=4$ copies of the original code which compose the enlarged code.}
	\label{fig:overlapping}
\end{figure}

Specifically, consider $C$ copies of our code, such that each independent noise process can act disjointly on precisely one copy. Specifically we will consider assigning each  $R_j$ to one of these copies; we will denote by $\lbrace \tilde R_j\rbrace_j$ these now-non-overlapping regions in the enlarged code (see \cref{fig:overlapping}). If we denote the error on each copy $c$ of our code by $E^{(c)}$, then the corresponding error on the original and enlarged codes are given by $E=\prod_c E^{(c)}$ and $\tilde E:=\otimes_c E^{(c)}$ respectively. The enlarged error $\tilde E$ can be considered a specific manifestation of the overall error $E$. 

Due to the non-overlapping nature of $\lbrace \tilde R_j\rbrace_j$, the enlarged errors $\tilde E$ are independently distributed,
\begin{align}
	\Pr(\tilde E)=\prod _j p_j\left( \!\left.\tilde E\right|_{\tilde R_j} \right).
\end{align}
This added redundancy means that the relationship between the errors in the original code $E$ and the enlarged code $\tilde E$ is, as noted earlier, one-to-many. To account for this, we need to introduce gauge generators acting between copies,
\begin{align}
\lbrace \tilde G_l\rbrace_l := \lbrace X_{i,c}\otimes X_{i,c+1}^{-1}, Z_{i,c}\otimes Z_{i,c+1}^{-1}   \rbrace_{i,c},
\end{align}
where $P_{i,c}$ denotes Pauli $P$ acting on the $i$th site in the $c$th copy of our code. By construction, the set of all enlarged errors $\tilde E$ which correspond to the same overall error $E$ is given by cosets of $\langle G_l\rangle_l$.

Next we define our stabilisers in the enlarged code, $\tilde S_k$, to simply be the original stabilisers $S_k$ acting on the the first copy of our code,
\begin{align}
	\tilde S_k:=S_k\otimes I\otimes \cdots \otimes I.
\end{align}
We now define the corresponding stat mech mapping analogous to that in \cref{sec:mapping}, as
\begin{align}
H_{\tilde E}(\vec c,\vec g)=-\!\!\sum_{j,\sigma\in\mathcal{P}_{\tilde R_j}} \!\!J_j(\sigma) \comm{\sigma}{\tilde E\prod_{k}\tilde S_k^{c_k}\prod_{l}\tilde G_l^{g_l}}\!\!,
\end{align}
with corresponding Nishimori condition
\begin{align}
\beta J_j(\sigma)=\frac{1}{|\mathcal P_{\tilde R_j}|}\sum_{\tau\in\mathcal P_{\tilde R_j}}\log p_j(\tau)\comm{\sigma}{\tau^{-1}}.
\end{align}
Recalling that the gauge generators $\lbrace G_l\rbrace_l$ map between different manifestations $\tilde E$ of the same overall error $E$, and $\lbrace \tilde S_k\rbrace_k$ between stabiliser-equivalent errors, we have that
\begin{align}
e^{-\beta H_{\tilde E}(\vec{0},\vec{0})} &= \Pr(\tilde E),\\
\sum_{\vec g}e^{-\beta H_{\tilde E}(\vec{0},\vec{g})} &= \Pr(E),\\
Z_{\tilde E}:=\sum_{\vec{c},\vec{g}}e^{-\beta H_{\tilde E}(\vec{c},\vec{g})} &= \Pr(\overline E),
\end{align}
in analogy to \cref{thm:fundie}.

\end{document}